\documentclass[a4paper,UKenglish,cleveref, autoref, thm-restate]{lipics-v2021}

\pdfoutput=1 

\newcommand{\shortversion}[1]{}

\hideLIPIcs  
\nolinenumbers 

\usepackage{amsmath,amssymb,amstext,mathtools}
\usepackage{booktabs}
\usepackage{makecell}
\usepackage{mdframed}
\usepackage{comment}
\usepackage[noadjust]{cite}
\usepackage[noend]{algorithm2e}
\newtheorem{thm}{Theorem}[section]
\newtheorem{lem}[thm]{Lemma}
\newtheorem{prop}[thm]{Proposition}

\newtheorem{clm}[thm]{Claim}
\newtheorem{defn}[thm]{Definition}
\newtheorem{Casex}[thm]{Case}

\crefname{thm}{Theorem}{Theorems}
\crefname{lem}{Lemma}{Lemmas}
\crefname{prop}{Proposition}{Propositions}
\crefname{clm}{Claim}{Claims}

\newcommand{\NP}{\ensuremath{\mathsf{NP}}}

\newcommand{\tw}{\textup{tw}}
\newcommand{\ETH}{\textup{ETH}}

\newcommand{\mwc}{\textup{\textsc{Multiway Cut}}}
\newcommand{\mc}{\textup{\textsc{Multicut}}}
\newcommand{\mcut}[1]{\textup{\textsc{Multicut}(\ensuremath{#1})}}

\newcommand{\gcut}{\textup{\textsc{Group 3-Terminal Cut}}}

\newcommand{\calH}{\mathcal{H}}

\newcommand{\mcH}{\mcut{\calH}}

\newcommand{\go}{\vec{g}}

\newcommand{\pl}{\pi}
\newcommand{\planar}{\textup{\sf planar}}
\newcommand{\kplanararg}[1]{\planar+#1 e}
\newcommand{\kplanar}{\kplanararg{\pl}}

\title{
	Multicut Problems in Almost-Planar Graphs:\\
	The Dependency of Complexity on the Demand Pattern\footnote{Most results contained in this article have been announced at ESA 2025.}
        }
\titlerunning{Multicut Problems in Almost-Planar Graphs} 

\authorrunning{ F.H\"orsch and D.Marx} 

\author{Florian H\"orsch}{CISPA Helmholtz Center for Information Security, Germany}{florian.hoersch@cispa.de}{https://orcid.org/0000-0002-5410-613X}{}
\author{D\'aniel Marx}{CISPA Helmholtz Center for Information Security, Germany}{marx@cispa.de}{https://orcid.org/0000-0002-5686-8314}{}

\Copyright{Florian H\"orsch and D\'aniel Marx} 



\usepackage[normalem]{ulem}

\newcommand{\omitnow}[1]{}

\begin{document}

\maketitle
\begin{abstract}
  Given a graph $G$, a set $T$ of terminal vertices, and a demand graph $H$ on $T$, the \textsc{Multicut} problem asks for a set of edges of minimum weight that separates the pairs of terminals specified by the edges of $H$. 
  The \textsc{Multicut} problem can be solved in polynomial time if the number of terminals and the genus of the graph is bounded (Colin de Verdi{\`{e}}re [\textit{Algorithmica, 2017}]). Restricting the possible demand graphs in the input leads to special cases of \textsc{Multicut} whose complexity might be different from the general problem.

  Focke et al.~[SoCG 2024] systematically characterized which special cases of \textsc{Multicut} are fixed-parameter tractable parameterized by the number of terminals on planar graphs. Moreover, extending these results beyond planar graphs, they precisely determined how the parameter genus influences the complexity and presented partial results of this form for graphs that can be made planar by the deletion of $\pi$ edges. Continuing this line of work, we complete the picture on how this parameter $\pi$ influences the complexity of different special cases and precisely determine the influence of the crossing number, another parameter measuring closeness to planarity.

  Formally, let $\mathcal{H}$ be any class of graphs (satisfying a mild closure property) and let $\textsc{Multicut}(\mathcal{H})$ be the special case when the demand graph $H$ is in $\mathcal{H}$. Our first main result is showing that if $\mathcal{H}$ has the combinatorial property of having \emph{bounded distance to extended bicliques}, then $\textsc{Multicut}(\mathcal{H})$ on \emph{unweighted} graphs is FPT parameterized by the number $t$ of terminals and $\pi$.
For the case when $\mathcal{H}$ does not have this combinatorial property,
    Focke et al.~[SoCG 2024] showed that  $O(\sqrt{t})$ is essentially the best possible exponent of the running time; together with our result, this gives a complete understanding of how the parameter $\pi$ influences complexity on unweighted graphs. 
    Our second main result is giving an algorithm whose existence shows that that the parameter crossing number behaves analogously if we consider $\textsc{Multicut}(\mathcal{H})$ on \emph{weighted} graphs.\end{abstract}

\section{Introduction}\label{sec:intro}

Most NP-hard problems remain NP-hard when restricted to planar graphs, unless they become trivial or irrelevant on planar graphs. There are very few NP-hard problems that are polynomial-time solvable on planar graphs for some nontrivial and interesting reason. One such case is the \mwc\ problem: given a graph $G$ and a set $T\subseteq V(G)$ of \emph{terminals}, the task is to find a \emph{multiway cut} $S$ of minimum weight, that is, a set $S\subseteq E(G)$ such that every component of $G\setminus S$ contains at most one terminal. Dahlhaus et al.~\cite{DBLP:journals/siamcomp/DahlhausJPSY94} showed that $\mwc$ in general graphs is \NP-hard\ even for three terminals, while it is polynomial-time solvable on planar graphs for every fixed number $|T|$ of terminals. The exponent of $n$ in the algorithm of Dahlhaus et al.~\cite{DBLP:journals/siamcomp/DahlhausJPSY94} is $O(|T|)$. Later, the dependence was improved to $O(\sqrt{|T|})$ \cite{DBLP:conf/icalp/KleinM12,ECDV}, which was shown to be optimal under the Exponential-Time Hypothesis (ETH) \cite{Marx12,Cohen-AddadVMM21}. 
The results were extended to bounded-genus graphs, again with an essentially best possible form of the exponent \cite{Cohen-AddadVMM21,ECDV}.

\mc\ is a generalization of \mwc\ where not all terminals need to be separated from each other, but the input specifies which pairs of terminals need to be separated (and we do not care whether other pairs are separated or not). \omitnow{We use the following formalization of the problem:}

\begin{center}

\begin{tabular}{|rl|}
  \hline
  \mc &\\ 
  {\bf Input:}& An edge-weighted graph $G$ 
                together with a set $T\subseteq V(G)$ of \emph{terminals};\\
  &and a \emph{demand graph} $H$ with $V(H)=T$.\\
		{\bf Output:} & 
		A minimum-weight set $S\subseteq E(G)$ such that $u$ and $v$ are in distinct components\\ & of $G\setminus S$ whenever $uv$ is an edge of $H$.\\
                \hline
              \end{tabular}
            \end{center}
Let us observe that \mwc\ is the special case of \mc\ when $H$ is a complete graph on $T$.
Colin de  Verdi{\`{e}}re \cite{ECDV} presented an algorithm for \mc\ where the exponent of the running time is $O(\sqrt{|T|})$, similarly to \mwc.
However, it is no longer clear anymore that this algorithm is a tight, optimal result and cannot be improved in some cases: it could be that there are special cases, special type of demand graphs where significantly better algorithms are possible. For example, if $H$ is a biclique (complete bipartite graph), then $\mc$ can be solved in polynomial time by a reduction to a minimum weight $s-t$ cut problem. A very natural question to consider is the case when $H$ is a complete 3-partite graph. This special case is the \gcut\ problem: given a graph $G$ with three sets of terminals $T_1$, $T_2$, $T_3$, the task is to find a set $S$ of edges of minimum total weight such that $G\setminus S$ has no path between any vertex of $T_i$ and any vertex of $T_j$ for all distinct $i,j\in [3]$. While a general \mc\ algorithm solves \gcut\ with the exponent of the running time being the square root of the number of terminals, the lower bounds for \mwc\ do not imply that this dependence is optimal.

            Focke et al.~\cite{focke2025journal} systematically analyzed which types of demand graphs make the problem more tractable. The question was explored in the following setting. Let $\calH$ be a class of graphs (i.e., a set of graphs closed under isomorphism). Then we denote by \mcH\ the special case of \mc\ that contains only instances $(G,H)$ where $H\in \calH$. For example, if $\calH$ is the class of cliques, then \mcH\ is exactly \mwc.

Ideally, for every class $\calH$ of graphs, we would like to understand the best possible running time of \mcH\ on planar graphs. In order to make the question more robust, Focke et al.~\cite{focke2025journal} considered only classes $\calH$ satisfying a mild closure property. We say that $H'$ is a \emph{projection} of $H$ if $H'$ can be obtained from $H$ by deleting vertices and identifying pairs of independent vertices. Moreover, a class of graphs $\calH$ is \emph{projection-closed} if, whenever $H\in \calH$ and $H'$ is a projection of $H$, then $H'\in\calH$ also holds. The intuition is that an instance having demand graph $H'$ can be easily simulated by an instance having demand graph $H$, thus it is reasonable to assume we only consider classes of demand graphs closed under this operation.

Let us consider some easy cases first. As we mentioned before, the \mc\ problem can be solved in polynomial time by a reduction to $s-t$ minimum cut if the demand graph is a biclique. Furthermore, if $H$ has only two edges, then a similar reduction is known~\cite{Hu63,DBLP:journals/jct/Seymour79a}. Isolated vertices of $H$ clearly do not play any role in the problem. We say that $H$ is a \emph{trivial pattern} if, after removing isolated vertices, it is either a biclique or has at most two edges.
An \emph{extended biclique} is a graph that consists of a complete bipartite graph together with a set of isolated vertices. 
Let $\mu$ be the minimum number of vertices that need to be deleted to make a graph $H$ an extended biclique. Then we say that $H$ has distance $\mu$ to extended bicliques. Moreover, we say that a graph class $\calH$ has \emph{bounded distance to extended bicliques} if there is a constant $\mu$ such that every $H\in\calH$ has distance at most $\mu$ to extended bicliques.


For planar graphs, Focke et al.~\cite{focke_et_al:LIPIcs.SoCG.2024.57} give the following characterization:\omitnow{ of the complexity of \mcH:}

\begin{thm}[\cite{focke_et_al:LIPIcs.SoCG.2024.57}]\label{thm:mainonlyplanar}
  Let $\calH$ be a computable projection-closed class of graphs. Then the following holds for edge-weighted \mcH\ on planar graphs.
  \begin{enumerate}
  \item If $\calH$ has bounded distance to extended bicliques, then there is an $f(t)n^{O(1)}$ time algorithm.
    \item Otherwise,
    \begin{enumerate}
    \item There is an $f(t)n^{O(\sqrt{t})}$ time algorithm.
    \item Assuming $\ETH$, there is a universal constant $\alpha>0$ such that for any fixed choice of $t\ge 3$, there is no $O(n^{\alpha \sqrt{t}})$ algorithm, even when restricted to unweighted instances with at most $t$ terminals.
    \end{enumerate}
  \end{enumerate}
\end{thm}

That is, depending on whether $\calH$ has bounded distance to extended bicliques, the number $t$ of terminals has to appear in the exponent. Focke et al.~\cite{focke2025journal} precisely characterized how the complexity changes if we move from planar graphs to bounded-genus graphs. As we can see, the parameter Euler genus $g$ can appear in the exponent in different ways, depending on $\calH$.

\begin{thm}[\cite{focke2025journal}]\label{thm:maingenus}
  Let $\calH$ be a computable projection-closed class of graphs. Then the following holds for edge-weighted \mcH.
  \begin{enumerate}
  \item If every graph in $\calH$ is a trivial pattern, then there is a  polynomial time algorithm.
  \item Otherwise, if $\calH$ has bounded distance to extended bicliques, then
    \begin{enumerate}
       \item There is an $f(g,t)n^{O(g)}$ time algorithm.
      \item Assuming $\ETH$, there is a universal constant $\alpha>0$ such that for any fixed choice of $\go\ge 0$, there is no $O(n^{\alpha (\go+1)/\log(\go+2)})$ algorithm, even when restricted to unweighted instances with orientable genus at most $\go$ and $t=3$ terminals.
      \end{enumerate} 
    \item Otherwise,
    \begin{enumerate}
    \item There is an $f(g,t)n^{O(\sqrt{g^2+gt+t})}$ time algorithm.
    \item Assuming $\ETH$, there is a universal constant $\alpha>0$ such that for any fixed choice of $\go\ge 0$ and $t\ge 3$, there is no $O(n^{\alpha \sqrt{\go^2+\go t+t}/\log(\go+t)})$ algorithm, even when restricted to unweighted instances with orientable genus at most $\go$ and at most $t$ terminals.
    \end{enumerate}
  \end{enumerate}
\end{thm}


In an earlier conference version, Focke et al.~\cite{focke_et_al:LIPIcs.SoCG.2024.57} also considered a much more restricted extension of planar graphs than bounded-genus graphs: graphs that can be made planar by deleting a few edges. We use $\kplanar$ for the class of graphs that can be made planar by removing at most $\pl$ edges. Given an instance $(G,H)$, we use $\pl$ for the smallest integer such that $G$ is in $\kplanar$.

Observe that in Theorem~\ref{thm:maingenus}, the algorithmic results are stated for the edge-weighted version, while the lower bounds are stated for the unweighted version. In fact, edge-weights do not make much difference in Theorem~\ref{thm:maingenus}, as polynomial edge weights can be simulated by parallel edges without changing the genus of the graphs. However, this simulation can easily change the number $\pl$ of edges that needs to be deleted to make the graph planar. This means that if we consider the influence of the parameter $\pl$ on the running time, the edge-weighted and the unweighted problem may behave differently. Focke et al.~\cite{focke2025journal} \omitnow{completely }characterized the influence of this parameter in the edge-weighted case:

\begin{thm}[\cite{focke2025journal}]\label{thm:mainplanarweighted}
  Let $\calH$ be a computable projection-closed class of graphs. Then the following holds for edge-weighted \mcH.
  \begin{enumerate}
  \item If every graph in $\calH$ is a trivial pattern, then there is a polynomial-time algorithm.
  \item Otherwise, if $\calH$ has bounded distance to extended bicliques, then
    \begin{enumerate}
    \item There is an $f(\pl,t)n^{O(\sqrt{\pl})}$ time algorithm.
    \item Assuming $\ETH$, there is a universal constant $\alpha>0$ such that for any fixed choice of $\pl\ge 0$, there is no $O(n^{\alpha \sqrt{\pl}})$ algorithm, even when restricted to instances in $\kplanar$ with $t=3$ terminals.
    \end{enumerate}
  \item Otherwise,
    \begin{enumerate}
    \item There is an $f(\pl,t)n^{O(\sqrt{\pl+t})}$ algorithm.
    \item Assuming $\ETH$, there is a universal constant $\alpha>0$ such that for any fixed choice of $\pl\ge 0$ and $t\ge 3$, there is no $O(n^{\alpha \sqrt{\pl+t}})$ algorithm, even when restricted to $\kplanar$ instances with at most $t$ terminals. 
    \end{enumerate}
  \end{enumerate}
\end{thm}

Observe that $\pl$ has a different form of influence on the exponent compared to genus: it is $O(\sqrt{\pl})$ instead of $O(g)$ in the case of bounded distance to extended bicliques, and $O(\sqrt{\pl+t})$ instead of $O(\sqrt{g^2+gt+t})$ in the unbounded distance case. Intuitively, having $\pl$ extra edges is a much more restricted extension of planar graphs than having genus $g$\omitnow{(i.e., having $g$ extra handles)}, and consequently, the parameter $\pl$ has a much weaker influence on the exponent than genus has.

For unweighted graphs, Focke et al.~\cite{focke_et_al:LIPIcs.SoCG.2024.57} showed that the problem is easier than the edge-weighted version and the parameter $\pl$ can be removed from the exponent.

\begin{thm}[\cite{focke_et_al:LIPIcs.SoCG.2024.57}]
	\label{kplan}\label{thm:unweightedmain}
		Unweighted \mc\ can be solved in time $f(\pl,t)n^{O(\sqrt{t})}$.
\end{thm}
However, this result leaves open the question what happens in unweighted \mcH\ if $\calH$ has bounded distance to extended bicliques. Based on Theorems~\ref{thm:maingenus} and \ref{thm:mainplanarweighted}, one would expect that it is possible to remove the number $t$ of terminals from the exponent as well.

\paragraph*{Our contributions}
First, we resolve the remaining question about unweighted \mcH\ and show that if $\calH$ has bounded distance to extended bicliques, then both $t$ and $\pl$ can be removed from the exponent\omitnow{ of the running time}. Recall that $\mu$ is the number of vertices needed to be deleted to make the demand graph an extended biclique. Our first main result is the following algorithm.
\begin{restatable}{thm}{extraedges}
\label{extraedges}
Unweighted \mc\ can be solved in time $f(\pl,t)n^{O(\sqrt{\mu})}$.
\end{restatable}

Observe that, in contrast to Theorem~\ref{thm:unweightedmain}, only $\mu$ is in the exponent of the running time of the algorithm of \cref{extraedges}. As $\mu\le t$, this algorithm gives an independent proof of Theorem~\ref{thm:unweightedmain} with a very different proof technique. With \cref{extraedges} at hand, we can complete the picture for the complexity of unweighted \mcH:

\begin{restatable}{thm}{mainunkplan}\label{thm:mainunkplan}
  Let $\calH$ be a computable projection-closed class of graphs. Then the following holds for unweighted \mcH.
  \begin{enumerate}
  \item If $\calH$ has bounded distance to extended bicliques, then there is an $f(\pl,t)n^{O(1)}$-time algorithm.
  \item Otherwise,
    \begin{enumerate}
    \item There is an $f(\pl,t)n^{O(\sqrt{t})}$ algorithm.
    \item Assuming $\ETH$, there is a universal constant $\alpha>0$ such that for any fixed choice of $t\ge 3$, there is no $O(n^{\alpha \sqrt{t}})$-time algorithm, even when restricted to planar instances with at most $t$ terminals. 
    \end{enumerate}
  \end{enumerate}
\end{restatable}
\medskip

We consider yet another, even more restricted parameter measuring the distance from planar graphs: the \emph{crossing number}, the minimum number of crossings needed when drawing the graph in the plane. Clearly, if a graph has a drawing with $cr$ crossings, then we can remove a set of at most $cr$ edges to make the graph planar. Thus the crossing number $cr$ is always at least $\pl$, that is, an algorithm parameterized by $\pl$ immediately gives an algorithm parameterized by $cr$.

Again, there is no straightforward reduction from the edge-weighted version to the unweighted version without increasing the crossing number. Thus both versions of the problem when parameterizing by crossing number are interesting. Clearly, the unweighted version is no harder than the edge-weighted one. Moreover, both versions parameterized by crossing number are no harder than their respective versions parameterized by $\pl$. But which of these versions are \emph{strictly} easier than the others?

It turns out that, when parameterizing by crossing number, the edge weights do not change the hardness of the problem and  both the edge-weighted and the unweighted versions parameterized by $cr$ have the same complexity as the \emph{unweighted} version parameterized by $\pl$. The following algorithm is our main result on the crossing number parameterization.

\begin{restatable}{thm}{maincr}\label{maincr}
Edge-weighted \mc\ can be solved in time $f(cr,t)n^{O(\sqrt{\mu})}$.
\end{restatable}
\medskip

Together with \cref{thm:mainonlyplanar}, it implies the following complete classification for the complexity of \mcH\ when parameterizing by $cr$:

\begin{restatable}{thm}{maincrossnumber}\label{thm:maincrossnumber}
  Let $\calH$ be a computable projection-closed class of graphs. Then the following holds for edge-weighted \mcH.
  \begin{enumerate}
  \item If $\calH$ has bounded distance to extended bicliques, then there is an $f(cr,t)n^{O(1)}$-time algorithm.
  \item Otherwise,
    \begin{enumerate}
    \item There is an $f(cr,t)n^{O(\sqrt{t})}$ algorithm.
    \item Assuming $\ETH$, there is a universal constant $\alpha>0$ such that for any fixed choice of $t\ge 3$, there is no $O(n^{\alpha \sqrt{t}})$-time algorithm, even when restricted to unweighted planar instances with at most $t$ terminals. 
    \end{enumerate}
  \end{enumerate}
\end{restatable}
\medskip

Note that the algorithmic results of Theorems~\ref{thm:mainunkplan} and \ref{thm:maincrossnumber} are incomparable: Theorem~\ref{thm:mainunkplan} considers only unweighted graphs, while Theorem~\ref{thm:maincrossnumber} considers edge-weighted graphs, but under parameterization by a potentially much larger parameter.

Theorems~\ref{thm:maingenus}--\ref{thm:maincrossnumber} complete the picture of how different parameters measuring closeness to planarity influence the running time of \mcH\ for different classes $\calH$ (see Figure~\ref{fig:maintable}). Let us observe that the exact form of the influence is very different for the different parameters. As we explain below, the treewidth of the so-called multicut dual of the optimum solution is the main factor determining the exponent of the running time \cite{ECDV,focke_et_al:LIPIcs.SoCG.2024.57}. For a formal definition of multicut duals, see \cref{xstrzdtufzigulh}. The multicut dual is an embedded graph with at most $t$ faces, thus the exponent $O(\sqrt{t})$ for planar graphs comes from the fact that a planar graph with $O(t)$ vertices (or faces) has treewidth $O(\sqrt{t})$. Intuitively, for nonplanar graphs, the situation changes as follows:

\begin{itemize}
\item \textbf{Genus.} If the genus is large, then a graph with $t$ vertices can have treewidth larger than $O(\sqrt{t})$. The maximum possible treewidth is given by the (nonobvious) formula $O(\sqrt{g^2+gt+t})$. Therefore, the maximum treewidth of the multicut dual is $O(\sqrt{g^2+gt+t})$, giving a bound on the exponent of the running time.
\item \textbf{Edge-weighted $\kplanar$.} With extra edges having infinite (or very large) weight, we can simulate multiple terminals at different parts of a planar graph.
  More precisely, a planar instance of \gcut\ with $t$ terminals in each of the three groups can be reduced to an instance of \mwc\ with three terminals if we connect the terminals in each group with a total of $\pl=3(t-1)$ infinite-weight edges, creating a $\kplanar$ instance. 
  Thus these extra edges influence the running time exactly as extra terminals would do. 
\item \textbf{Unweighted $\kplanar$.}
  We can always try to solve the problem by putting the $\pl$ extra edges into the solution, removing them from $G$, and then solving the resulting planar instance. Of course, this might not always lead to an optimum solution, but can be up to $\pl$ edges larger. However, this argument allows us to say that if the solution contains a cut that is ``inefficient'' in the sense that it is larger then the minimum cut by more than $\pl$, then the cut can be replaced by a minimum cut and the $\pl$ extra edges, resulting in a srictly smaller solution. Using these type of arguments, the problem can be reduced to a planar setting where the number of faces of the multicut dual does not depend on $\pl$.

  \item \textbf{Crossing number.} Each crossing can increase the number of vertices of the multicut dual by one, thus the number of vertices and the treewidth of the multicut dual can be $O(t+cr)$ and $O(\sqrt{t+cr})$, respectively. However, all these extra vertices of the multicut dual are at fixed places determined by the crossing. As pointed out by Focke et al.~\cite{focke2025journal} it is sufficient to consider the treewidth of the multicut dual \emph{after removing all the vertices whose locations are fixed} (see \cref{algogenext}). Indeed, after removing these vertices, the treewidth of the multicut dual does not depend on the crossing number.
  
\end{itemize}

\begin{figure}

\newcommand{\exphigh}[1]{\textcolor{red}{#1}}

\renewcommand{\arraystretch}{1.4}
\begin{tabular}{c@{\hskip 8pt}c@{\hskip 8pt}c@{\hskip 8pt}c@{\hskip 8pt}c}
  &  &weighted&unweighted&crossing\\[-2mm]
  &genus
    & $\kplanar$
    & $\kplanar$
    & number \\
    \toprule
    \addlinespace
  $\calH$ with &&&&\\[-2mm]
  bounded distance
    & $f(g,t)n^{\exphigh{O(g)}}$ 
    & $f(\pl,t)n^{\exphigh{O(\sqrt{\pl})}}$ 
    & $f(\pl,t)n^{\exphigh{O(1)}}$ (*) 
    & $f(cr,t)n^{\exphigh{O(1)}}$ (*) \\[-2mm]
    to extended bicliques&&&&\\
    \addlinespace
    $\calH$ with &&&&\\[-2mm]unbounded distance
    & $f(g,t)n^{\exphigh{O(\sqrt{g^2+gt+t})}}$ 
    & $f(\pl,t)n^{\exphigh{O(\sqrt{\pl+t})}}$ 
    & $f(\pl,t)n^{\exphigh{O(\sqrt{t})}}$ 
    & $f(cr,t)n^{\exphigh{O(\sqrt{t})}}$ (*)\\[-2mm]
  to extended bicliques &&&&\\
  \bottomrule
\end{tabular}
\caption{The running time of the algorithms for \mcH\ with the different parameterizations. The new results of the paper are marked by (*). Note that, for every fixed projection-closed class $\calH$, the exponents of the running time are tight (up to logarithmic factors).}\label{fig:maintable}
  
\end{figure}

In the light of these results, it would be interesting to study further notions of almost-planar graphs. We here mention two further notions of almost-planar graphs that we are \emph{not} studying in this paper and argue why studying them does not promise to be very fruitful.

 First, it would be very natural to consider graphs that can be made planar by the deletion of $\pl$ \emph{vertices}. However, the problem can be very hard even for small values of this parameter. Specifically, the \mwc\ problem with three terminals is \NP-hard even for graphs that can be made planar by the deletion of three vertices, as the \gcut\ problem on planar graphs can be reduced to such an instance \cite{DBLP:journals/siamcomp/DahlhausJPSY94}. Thus we cannot expect results similar to Theorems~\ref{thm:maingenus}--\ref{thm:maincrossnumber} for this parameter.

Further, a graph is \emph{1-planar} if it can be drawn in the plance such that every edge participates in at most one crossing. Every graph can be made 1-planar by subdividing each edge a sufficient number of times. Subdividing an edge (and giving the same weight to both new edges) does not change the \mc\ problem, hence the problem is as hard on 1-planar graphs as on general graphs. Therefore, 1-planar graphs are irrelevant for the study of \mc\ and its special cases.

The remainder of this article is organized as follows. In \cref{prelsec}, we give some preliminaries we need for our algorithms. Afterwards, we prove \cref{extraedges} and \cref{maincr} in \cref{edgedel} and \cref{crosssec}, respectively.

\section{Preliminaries}\label{prelsec}
We here give some preliminaries for our article. We first give some notation and then review the concept of multicut duals.

\subsection{Notation}


For some positive integer $k$, we use $[k]$ for $\{1,\ldots,k\}$. A {\it subpartition} of a set $V$ is a collection $\mathcal{P}$ of disjoint subsets of $V$. We use $\bigcup \mathcal{P}$ for the union of all classes of $\mathcal{P}$. Moreover, $\mathcal{P}$ is a partition if $\bigcup \mathcal{P}=V$.

A {\it graph} $G$ consists of a finite set $V(G)$ whose elements are called {\it vertices} and a set $E(G)$ of subsets of $V(G)$ of size 2 called {\it edges}. For some $E' \subseteq E(G)$, we use $V(E')$ for the set of vertices contained in at least one edge in $E'$. Further, we use $G\setminus E'$ for the unique graph on $V(G)$ with $E(G\setminus E')=E(G)\setminus E'$. Next, for some $V'\subseteq V(G)$, we use $G[V']$ for the subgraph of $G$ induced by $V'$ that is, $V(G[V'])=V'$ and $E(G[V'])$ contains all edges of $E(G)$ fully contained in $V'$. We also use $G-V'$ for $G[V(G)\setminus V']$. Given a graph $G$ and two disjoint sets $X,Y\subseteq V(G)$, we use $\delta_{G}(X,Y)$ for the set of edges in $E(G)$ containing a vertex of $X$ and a vertex of $Y$. We further use $d_G(X,Y)$ for $|\delta_{G}(X,Y)|, \delta_G(X)$ for $\delta_G(X,V(G)\setminus X)$, and $d_G(X)$ for $|\delta_{G}(X)|$. We say that $G$ is {\it connected} if $d_G(X)\geq 1$ for every nonempty, proper $X \subseteq V(G)$. Next, we say that a graph $G'$ is a {\it subgraph} of $G$ if $V(G')\subseteq V(G)$ and $E(G')\subseteq E(G)$. A {\it component} of $G$ is a maximal connected subgraph of $G$. 

A graph $P$ is a {\it $u_1u_2$-path} for some $u_1,u_2\in V(P)$ if $P$ is connected, $d_P(u_i)=1$ holds for all $i \in [2]$  and $d_P(v)=2$ holds for all $v \in V(P)$. We say that the {\it length} of $P$ is $|E(P)|$. If $u_2 \in U$ for some $U \subseteq V(G)\setminus \{u_1\}$, we say that $P$ links $u_1$ and $U$.  For a graph $G$ and a positive integer $k$, we say that a set $X \subseteq V(G)$ is a {\it $k$-dominating} if for every $u \in V(G)$, there exists a path of length at most $k$ linking $u$ and $X$.

Throughout the article, we use the notion of treewidth. However, for the understanding of the current article, the exact definition of treewidth, which can for example be found in \cite{cfklmpps}, is not required. The properties described in \cref{domplanar} and \cref{algogenext} as well as the well-known fact that the treewidth of a graph is the maximum of the treewidth of one of its components will be sufficient.

  Recall that an instance of \mc{} consists of a graph $G$ and a graph $H$ with $V(H)\subseteq V(G)$ and a {\it multicut} for $(G,H)$ is a set $S \subseteq E(G)$ such that $t_1$ and $t_2$ are in distinct components of $G \setminus S$ for all $t_1,t_2 \in V(H)$ with $t_1t_2 \in E(H)$.
For the edge-weighted version of the problem, we assume that the input graph $G$ is equipped with a weight function $w:E(G)\to \mathbb{Z}^+$. A {\it weighting} of a graph is the graph together with a weight function.

An {\it extended biclique decomposition} of a graph $H$ is a partition $(B_1,B_2,I,X)$ of $V(H)$ such that all vertices of $I$ are isolated in $H - X$, $E(H)$ contains the edge $t_1t_2$ for all $t_1 \in B_1$ and $t_2 \in B_2$, and $H[B_i]$ is an edgeless graph for $i \in [2]$. We use $\mathcal{H}_\mu$ for the set of graphs $H$ that admit an extended biclique decomposition $(B_1,B_2,I,X)$ with $|X|\leq \mu$. 

\subsection{Multicut duals}\label{xstrzdtufzigulh}
We now review the concept of so-called multicut duals, which will play a crucial role throughout the article. Recall that an instance $(G,H)$ of $\mc$ is \emph{planar} if $G$ is given in the form of a graph cellularly embedded in the plane.
Given a graph $C$ that is embedded in the plane and in general position with $G$, we denote by $e_G(C)$ those edges of $G$ that are crossed by at least one edge of $C$. The \emph{weight $w(C)$} of $C$ is defined to be $w(e_G(C))$ where $w$ is the weight function associated with $G$. A \emph{multicut dual} for $(G,H)$ is a graph $C$ embedded in the plane that is in general position with $G$ and has the property that, for every $t_1,t_2 \in V(H)$ with $t_1t_2 \in E(H)$, the terminals $t_1$ and $t_2$ are contained in different faces of $C$.

The following result is an immediate consequence of Lemmas 2.3 and 2.4 in \cite{focke2025journal}.

\begin{lem}\label{dualcut}
Let $(G,H)$ be a planar instance of $\mc$ and $C$ a minimum weight multicut dual for $(G,H)$. Then $e_G(C)$ is a minimum weight multicut for $(G,H)$.
\end{lem}

 In order to make use of the topology of certain multicut duals, we need two results related to treewidth. The first one of the following two results provides a sufficient criterion for the treewidth of certain graphs being bounded and the second one allows to exploit bounded treewidth algorithmically. 
 
 The following result can be obtained from Theorem 3.2 in \cite{10.1145/1077464.1077468} and the fact that the treewidth of any graph can only be larger than its branchwidth by a constant factor.

\begin{prop}\label{domplanar}
Let $C$ be a planar graph that admits a $5$-dominating set of size $k$ for some positive integer $k$. Then $\tw(C)= O(\sqrt{k})$.
\end{prop}

An \emph{extended} planar instance $(G,H,F^*)$ of \mc{} is a planar instance $(G,H)$ of \mc{} together with a set $F^*$ of faces of $G$ which is given with the input. Given a set $F^*$ of faces of $G$ and a multicut dual $C$ for $(G,H)$, let $C-F^*$ be the graph obtained from $C$ by removing every vertex that is in a face of $F^*$ and removing every edge that intersects a face of $F^*$.  The following result is the restriction of Theorem 1.11 in \cite{focke2025journal} to the planar setting.

\begin{restatable}{thm}{algogenext}\label{algogenext}
Let $(G,H,F^*)$ be an extended planar instance of \mc{} such that every subcubic inclusionwise minimal minimum-weight multicut dual $C$ satisfies $\tw(C-F^*)\leq  \beta$. Then an optimum solution of $(G,H)$ can be found in time $f(|F^*|,t)n^{O(\beta)}$.
\end{restatable}
Notice that Theorem~\ref{algogenext} requires a bound on the treewidth of every subcubic inclusionwise minimal minimum-weight multicut dual, and it is not sufficient to prove the existence of just one that satisfies the bound. This is an inherent feature of the proof technique, which modifies the multicut dual to obtain further properties.


\section{Edge-deletion distance to planarity}\label{edgedel}

This section is dedicated to proving \cref{extraedges}, which we restate here for convenience. 
\extraedges*
  Throughout this section, we say that an instance $(G,H)$ of $\mc$ is \emph{$\kplanar$} if a set $E_\pl \subseteq E(G)$ is given with the instance such that $|E_\pl|=\pl$ and $G\setminus E_\pl$ is planar. We further use $T$ for $V(H)$, $W$ for $V(E_\pl)$, and $G_0$ for $G\setminus E_\pl$. Also, for some $S\subseteq E(G)$, we use $S_0$ for $S\setminus E_\pl$. We say that $(G,H)$ is {\it connected} if $G$ is connected. We further let $(B_1,B_2,I,X)$ be an extended biclique decomposition of $H$ with $|X|=\mu$. Throughout this section, unless specified otherwise, we suppose that the instance is unweighted.

Before giving the detailed proof of \cref{extraedges}, which is rather long and technical, we now give an overview of the main steps of the proof.

The overall algorithmic strategy for Theorem \ref{extraedges} is to run an algorithm on the given instance that either produces a minimum multicut for the instance or computes a collection of instances for which the parameters in consideration are smaller and such that from minimum multicuts of all those instances, a minimum multicut for the original instance can be computed. We need the following definitions.

Given instances $(G,H)$ and $(G',H')$ of $\mc$,  $(G',H')$ is a {\it subinstance} of $(G,H)$ if $G'$ is a labelled subgraph of $G$ and $H'$ is a labelled induced subgraph of $H$. Next, an {\it extended subinstance} of $(G,H)$ consists of a set $S'\subseteq E(G)$ and a subinstance $(G',H')$ of $(G,H)$ with $E(G')\cap S'=\emptyset$. When we say that an extended subinstance $(S',(G',H'))$ has certain properties which can only be associated to an instance of \mc, we refer to $(G',H')$. Next, we say that an extended subinstance $(S',(G',H'))$ of $(G,H)$ is {\it optimumbound} for $(G,H)$ if $S' \cup S''$ is a minimum multicut for $(G,H)$ for every minimum multicut $S''$ for $(G',H')$.

The main technical difficulty is to prove the following result, which shows that the above described reduction to smaller instances exists for connected instances. With this result at hand, \cref{extraedges} follows readily. In particular, if disconnected instances are created, they can easily be reduced to instances corresponding to the components of the graph. Further observe that $\mu$ cannot increase  when moving to a subinstance as the new demand graph is an induced subgraph of the previous one.

\begin{restatable}{lem}{creins}
\label{creins}
 Let $(G,H)$ be a connected $\kplanar$ instance of $\mc$. Then, in $f(\pi,t)n^{O(\sqrt{\mu})}$, we can run an algorithm that returns a set $S\subseteq E(G)$  and a collection $(S_i,(G_i,H_i))_{i \in [k]}$ of extended subinstances of $(G,H)$  for some $k=f(\pi,t)$ such that for all $i \in [k]$, we have that $\pi(G_i)+|V(H_i)|<\pi(G)+|V(H)|$ holds and either $S$ is a minimum multicut for $(G,H)$ or there exists some $i \in [k]$ such that $(S_i,(G_i,H_i))$ is optimumbound for $(G,H)$.
  \end{restatable}
  
  We now give an overview of the proof of \cref{creins}. Let $(G,H)$ be a $\kplanar$ instance of \mc.
 The idea is to have a branching algorithm that gradually guesses more and more crucial information about a minimum multicut for $(G',H')$ while the number of directions it branches into remains bounded. The current status of this procedure is displayed by a subpartition of the vertices of $V(G)$ with some extra properties. We now give the definition of a state, which is the crucial tool to track the progress of our algorithm.
 \begin{defn}
 A {\it state} is a subpartition $\mathcal{P}$ of $V(G)$ with $W \cup T \subseteq \bigcup \mathcal{P}$ and $Y \cap (W \cup T)\neq \emptyset$ for all $Y \in \mathcal{P}$.
 \end{defn} Given a certain state, the objective is to understand whether there exists a minimum multicut $S$ for $(G,H)$ such that the structure of the components of $G_0\setminus S$ reflects the structure of $\mathcal{P}$.  In order to make this more formal, we need the following definitions: Given a set $V$ and two subpartitions $\mathcal{P}^1$ and $\mathcal{P}^2$ of $V$, we say that $\mathcal{P}^2$ {\it extends} $\mathcal{P}^1$ if $|\mathcal{P}^1|=|\mathcal{P}^2|$ and for every $Y \in \mathcal{P}^1$, there exists some $\overline{Y}\in \mathcal{P}^2$ with $Y \subseteq \overline{Y}$ such that 
  for any distinct $Y_1,Y_2 \in \mathcal{P}^1$, we have $\overline{Y}_1\neq \overline{Y}_2$.  Given a set $S\subseteq E(G)$, we use $S_0$ for $S \cap E(G_0)$ and $\mathcal{K}(G_0\setminus S_0)$ for the unique partition of $V(G)$ in which two vertices are in the same partition class if and only if they are in the same component of $G_0\setminus S_0$.  For a state $\mathcal{P}$, we say that a multicut $S$ for $(G,H)$ {\it respects} $\mathcal{P}$ if $\mathcal{K}(G_0\setminus S_0)$ extends $\mathcal{P}$ and we say that a state is {\it valid} if there exists a minimum multicut for $(G,H)$ respecting it. Further, $\mathcal{P}$ is {\it maximum valid} if $|\mathcal{P}|$ is maximum among all valid states.
 
 The idea of our algorithm now consists of considering a valid state, choosing a vertex that is not contained in any class of the state and making the classes larger by adding this vertex to one of them. As we do not know which class this vertex should be added to, we try all classes of the state and branch in these directions. The difficulty is now to carefully choose the vertex in consideration and make some additional modifications on the states so that this procedure terminates sufficiently fast.

In order to be able to explain how this works, we first need some more definitions.  Given disjoint sets $Y_1,Y_2\subseteq V(G)$, we use $\lambda_{G_0}(Y_1,Y_2)$ to denote the size of the smallest set of edges in $E(G_0)$ whose deletion from $G_0$ results in a graph in which there exists no component containing both a vertex of $Y_1$ and a vertex of $Y_2$. 
Next, for a state $\mathcal{P}$ and some $Y \in \mathcal{P}$, we define $\alpha_{G,\mathcal{P}}(Y)=\lambda_{G_0}(Y,\bigcup \mathcal{P}\setminus Y)-\lambda_{G_0}(Y\cap (W \cup T), (W \cup T) \setminus Y)$. Intuitively speaking, this measure describes how much more costly it is to separate the class from the other classes of the state in comparison to only separating the corresponding elements of $W \cup T$. Given a state $\mathcal{P}$, we now classify its classes according to their value of $\alpha$ and the value of $\alpha$ of their neighboring classes. More formally, given a state $\mathcal{P}$ and some $Y \in \mathcal{P}$, we say that $Y$ is {\it fat} in $\mathcal{P}$ if $\alpha_{G_0,\mathcal{P}}(Y)\geq \pi(t+1)+1$. For some $Y \in \mathcal{P}$ that is not fat, we say that $Y$ is {\it fat-neighboring} in $\mathcal{P}$ if there exists some fat $Y' \in \mathcal{P}$, some $u \in Y$, and some $v \in Y'$ such that $uv \in E(G_0).$ Finally, for some $Y \in \mathcal{P}$ which is neither fat nor fat-neighboring in $\mathcal{P}$, we say that $Y$ is {\it thin} in $\mathcal{P}$. We refer by $\mathcal{P}_{\textup{fat}}, \mathcal{P}_{\textup{f-n}}$, and $\mathcal{P}_{\textup{thin}}$ to the fat, fat-neighboring, and thin classes of $\mathcal{P}$, respectively. 

In order to make use of this distinction, we crucially need a structural property of minimum multicuts for $(G,H)$. Namely, the following result shows that, for any minimum multicut $S$ for $(G,H)$ and any vertex set of a component of $G_0\setminus S_0$, one of two things holds: either the value of $\alpha$ is not too large or a terminal of $X$ (of the extended biclique decomposition $(B_1,B_2,I,X)$ is close to this vertex set in a certain sense. In order to measure this proximity of certain vertex sets with respect to a multicut, we introduce the following notion:  Let $S\subseteq E(G)$, and $Y_1,Y_2 \subseteq V(G)$. Then we denote by $\textup{dist}_{G_0,S_0}(Y_1,Y_2)$ the smallest integer $q$ such that $G_0$ contains a $Y_1Y_2$-path $P$ with $|E(P)\cap S_0|=q$. We are now ready to state the following structural result for minimum multicuts.

\begin{restatable}{lem}{hauptred}\label{hauptred}
Let $(G,H)$ be a $\kplanar$ instance of $\mc$, let $S$ be a minimum multicut for $(G,H)$ and let $Y \in \mathcal{P}_{\textup{fat}}$ for $\mathcal{P}=\mathcal{K}(G_0\setminus S_0)$. Then  $\textup{dist}_{G_0,S_0}(Y,X)\leq 1$.
\end{restatable}

We now explain the strategy of our algorithm. Given a state with an extra property that will be explained later (``relevant state''), we choose a vertex that is not contained in any class of the state and has a neighbor in a thin class of the state in $G_0$. We then guess a partition class which this vertex will be added to and modify the state accordingly. In order to show the functionality of this algorithm, there are two points that need further explanation. First, we need to explain how to conclude when the described modification is no longer possible and second, we need to explain why this procedure terminates sufficiently fast.

We first explain the scenario that we can no longer execute the described modification, so there does not exist any vertex in the graph that is not contained in any class of the state $\mathcal{P}$ in consideration and that has a neighbor in a thin class of $\mathcal{P}$ in $G_0$. We call such a state {\it complete}, all other states are called {\it incomplete}. For complete states, we need to consider two different cases, depending on whether $\mathcal{P}$ contains a thin class.

The case that there exists a thin class $Y$ can be handled with elementary means. Observe that for this thin class, all its neighbors in $G_0$ are contained in a different class of $\mathcal{P}$. Hence, if $\mathcal{P}$ is valid, then all edges in $\delta_{G_0}(Y)$ are contained in a minimum multicut for $(G,H)$. We can hence delete these edges and solve the remaining instance, which turns out to be significatly smaller in a certain sense. This is summarized in the following result.

\begin{restatable}{lem}{phigross}\label{phigross}
Let $(G,H)$ be a connected $\kplanar$ instance of $\mc$ and let  a complete valid state $\mathcal{P}$ of $(G,H)$ with $\mathcal{P}_{\textup{thin}}\neq \emptyset$ be given. Then in polynomial time, we can compute an optimumbound extended subinstance $(S',(G',H'))$ of $(G,H)$ such that $\pi(G')+|V(H')|<\pi(G)+|V(H)|$.
\end{restatable} 

We now turn to the case that $\mathcal{P}$ does not contain a thin set. In order to handle this case, geometrical arguments on multicut duals will be needed. Namely, we show that if $\mathcal{P}$ is valid and does not contain a thin set, then the problem of finding a minimum multicut for $(G,H)$ can be reduced to solving a planar instance of multicut, which can be solved sufficiently fast due to \cref{algogenext} and an argument on the treewidth of the multicut duals for this instance. The treewidth bound is based on the fact that every face of a multicut dual is close to a face containing a terminal of $X$. Formally, we prove the following result, where $\tau(\mathcal{P})$ denotes the sum of the number of components of $G_0[Y]$ over all $Y \in \mathcal{P}$.

\begin{restatable}{lem}{compphiklein}\label{compphiklein}
Let $(G,H)$ be a connected $\kplanar$ instance of $\mc$ and let  a complete maximum valid state $\mathcal{P}$ of $(G,H)$ with $\mathcal{P}_{\textup{thin}}=\emptyset$ be given. Then a minimum multicut for $(G,H)$ can be computed in $f(\pi,\tau(\mathcal{P}))n^{O(\sqrt{\mu})}$.
\end{restatable} 

It remains to show that we can design our procedure so that we reach a complete state sufficiently fast. In order to keep track of the progress of our algorithm, we need a measure of how far advanced a given state is. To this end, given a state $\mathcal{P}$, we first define an integer $\kappa(\mathcal{P},Y)$  for every $Y \in \mathcal{P}$ by
\[
    \kappa(\mathcal{P},Y)= \left\{\begin{array}{lr}
        2(\pi(t+1)+1), & \text{for }Y \in \mathcal{P}_{\textup{fat}},\\
        \pi(t+1)+1+\alpha_{G_0,\mathcal{P}}(Y), & \text{for }Y \in \mathcal{P}_{\textup{f-n}},\\
        \alpha_{G_0,\mathcal{P}}(Y), & \text{for }Y \in \mathcal{P}_{\textup{thin}}.
        \end{array}\right\}
  \]
 Now, we define $\kappa(\mathcal{P})=\sum_{Y \in \mathcal{P}}\kappa(\mathcal{P},Y)$.

 Our objective is to make the value of $\kappa$ increase in each iteration of the algorithm. In order to make this possible, we need that the addition of a vertex to a given set in the state makes the value of $\alpha$ of this set increase. In order for that to hold, we need the classes in the state to be as large as possible in a certain sense. More precisely, given a graph $G$ and two disjoint sets $Y_1,Y_2\subseteq V(G)$, we say that a set $Y_3$ is {\it relevant} for $(Y_1,Y_2)$ in $G$ if $Y_1\subseteq Y_3, Y_2\cap Y_3=\emptyset, d_G(Y_3)=\lambda_G(Y_1,Y_2)$ and $Y_3$ is maximum among all sets with these properties. Further, we say that a state $\mathcal{P}$ is {\it relevant} if for all $Y \in \mathcal{P}$, we have that $Y$ is a relevant set for $(Y,\bigcup \mathcal{P}\setminus Y)$.
We can make a state $\mathcal{P}$ relevant by iteratively replacing each class $Y\in \mathcal{P}$  by the set which is relevant in this sense.  However, it needs careful verification that this process does not destroy any of the important properties of the state.

\begin{restatable}{lem}{makerelevant}\label{makerelevant}
Let a connected $\kplanar$ instance $(G,H)$ of $\mc$ and a state $\mathcal{P}$ of $(G,H)$ be given. Then, in polynomial time, we can compute a relevant state $\mathcal{P}'$ of $(G,H)$ such that $|\mathcal{P}'|=|\mathcal{P}|,\tau(\mathcal{P}')\leq \tau(\mathcal{P})$, $\kappa(\mathcal{P}')\geq \kappa(\mathcal{P})$ and, if $\mathcal{P}$ is maximum valid, then so is $\mathcal{P}'$.
\end{restatable}

With this result at hand, we can handle incomplete states. As pointed out before, given an incomplete relevant state $\mathcal{P}$, we choose a vertex $u_0 \in V(G)\setminus \bigcup\mathcal{P}$ that has a neighbor in a thin class $Y$ of $\mathcal{P}$. We now consider all states that are obtained by adding $u_0$ to one of these partition classes. If $u_0$ is added to a thin or a fat-neighboring partition class, then, as $\mathcal{P}$ is relevant, the value of $\alpha$ increases for this class and hence $\kappa$ increases. If $u_0$ is added to a fat class, then $Y$ becomes fat-neighboring and hence $\kappa$ also increases. This shows that a complete state can be reached sufficiently fast. In order to apply \cref{compphiklein}, we also need to make sure that $\tau(\mathcal{P})$ does not increase too fast. However, this follows readily by construction. Our handling of relevant incomplete states is summarized in the following result.

\begin{restatable}{lem}{dealincomplete}\label{dealincomplete}
Let $(G,H)$ be a connected $\kplanar$ instance of $\mc$ and let a relevant incomplete state $\mathcal{P}$ of $(G,H)$ be given. Then, in polynomial time, we can compute a collection $(\mathcal{P}^i)_{i \in [q]}$ of $q$ states of $(G,H)$ for some $q\leq 2\pi+t$ such that $\tau(\mathcal{P}^i)\leq \tau(\mathcal{P})+1$ and $\kappa(\mathcal{P}^i)>\kappa(\mathcal{P})$ hold for all $i \in [q]$ and if $\mathcal{P}$ is maximum valid, then there exists some $i \in [q]$ such that $\mathcal{P}^i$ is maximum valid.
\end{restatable}

The overall strategy for handling \cref{creins} is now to apply \cref{makerelevant} and \cref{dealincomplete} until we are left with a collection of complete states. These can then be handled by \cref{phigross} and \cref{compphiklein}.

The remainder of Section \ref{edgedel} can be summarized as follows. First, in Section \ref{prep}, we collect some simple preliminary results for our algorithm. Next, in Section \ref{struktur}, we give some structural results on minimum multicuts we need for our algorithm. In particular, we prove Lemma \ref{hauptred}. Next, in Section \ref{relevant}, we prove that arbitrary states can efficiently be made relevant, that is, we prove Lemma \ref{makerelevant}. After, in Section \ref{incompl}, we show how an incomplete relevant state can be replaced by a small number of states whose value for $\kappa$ is larger, meaning we prove Lemma \ref{dealincomplete}.
Next, in Section \ref{compl}, we show how to deal with complete states. More precisely, we prove Lemmas \ref{phigross} and \ref{compphiklein} that show how to obtain \cref{creins} once a complete valid state is available In Section \ref{reduce}, we prove Lemma \ref{creins} from the results obtained in the previous sections. Finally, in Section \ref{finalg}, we conclude Theorem \ref{extraedges} from \cref{creins}.

 \subsection{Preparatory results}\label{prep}

In this section, we collect some preliminary results we need for our algorithm. First, a solution of the multicut problem can be checked in polynomial time (e.g., with the use of breadth first search). 
 \begin{prop}\label{checkmc}
Let $(G,H)$ be an instance of $\mc$ and $S \subseteq E(G)$. Then we can check in polynomial time whether $S$ is a multicut for $(G,H)$.
\end{prop}

We next need the following simple property which holds for optimization problems aiming for the minimization of the number of elements we delete in general. Its proof can be found in \cite{focke_et_al:LIPIcs.SoCG.2024.57}.
\begin{prop}\label{cutequiv}
Let $(G,H)$ be an instance of $\mc$ with $S$ being a minimum multicut for $(G,H)$ and $F \subseteq S$. Then for every minimum multicut $S'$ for $(G\setminus F,H)$, we have that $S' \cup F$ is a minimum multicut for $(G,H)$.
\end{prop}

The next two results we need deal with the number of edges whose deletion is needed to make a given graph planar. 
We first need the following result, showing that a small edge set whose deletion makes a graph planar can be found sufficiently fast for our purposes. Analogous results are known also for the vertex-deletion version \cite{DBLP:journals/algorithmica/MarxS12,DBLP:conf/soda/AdlerGK08,DBLP:conf/soda/JansenLS14}, which can be made to work also for the edge-deletion version by appropriate reductions.

\begin{prop}[\cite{DBLP:conf/stoc/KawarabayashiR07}]\label{findmod}
Let $G$ be a graph and $\pl$ a nonnegative integer. Then there exists an algorithm that either outputs a set $E_\pl\subseteq E(G)$ such that $|E_\pl|\leq \pl$ and $G\setminus E_\pl$ is planar or correctly reports that no such set exists and runs in $f(\pl)n$. 
\end{prop}

The next lemma is important as it shows that progress is made on an instance when an edge linking two connected components of the planar part of the graph is detected. Its proof can be found in \cite{focke_et_al:LIPIcs.SoCG.2024.57}.
\begin{lem}\label{conn}
Let $G$ be a graph and $E'\subseteq E(G)$ be such that $G\setminus E'$ is planar. If there is an edge $e$ in $E'$ whose endvertices are contained in distinct components of $G\setminus E'$, then $G\setminus (E'\setminus \{e\})$ is planar.
\end{lem}
Further, we give some basic results on relevant sets. For the first of these results, the uniqueness of $Y_3$ is well-known and can be established by a simple submodularity argument. Moreover, computing $Y_3$ is possible through a simple modification of the algorithm of Edmonds and Karp and the last property follows directly from the minimality of $Y_3$. Propositions \ref{relev2} and \ref{relev3} follow directly from the definition of relevant sets.

\begin{prop}\label{relev1}
Let $G$ be a connected graph and $Y_1,Y_2\subseteq V(G)$ be disjoint sets. Then there exists a unique relevant set $Y_3$ for $(Y_1,Y_2)$. Moreover $Y_3$ can be computed in polynomial time and every component of $G[Y_3]$ contains a component of $G[Y_1]$. 
\end{prop}

\begin{prop}\label{relev2}
Let $G$ be a graph and $Y_1,Y_2\subseteq V(G)$ be disjoint, $Y_3$ the relevant set for $(Y_1,Y_2)$ and $Y_4 \subseteq V(G)$ with $Y_1\subseteq Y_4 \subseteq Y_3$. Then $Y_3$ is the relevant set for $(Y_4,Y_2)$.
\end{prop}

\begin{prop}\label{relev3}
Let $G$ be a graph and $Y_1,Y_2\subseteq V(G)$ be disjoint, $Y_3$ the relevant set for $(Y_1,Y_2)$ and $Y_4 \subseteq V(G)$ with $Y_2\subseteq Y_4 \subseteq V(G)\setminus Y_3$. Then $Y_3$ is the relevant set for $(Y_1,Y_4)$.
\end{prop}

Throughout our algorithm, it is important to  assure that the value of $\kappa$ increases continuously. The following easy result establishes a criterion to test whether this is actually the case. 

\begin{prop}\label{kappaincrease}
Let $(G,H)$ be a $\kplanar$ instance of \mc, and $\mathcal{P}^1$ and $\mathcal{P}^2$ two states for $(G,H)$ such that $\mathcal{P}^2$ extends $\mathcal{P}^1$. Then $\kappa(\mathcal{P}^2)\geq \kappa(\mathcal{P}^1)$. Moreover, if there exists some $Y_0 \in \mathcal{P}^1$ and $\overline{Y}_0 \in \mathcal{P}^2$ with $Y_0 \subseteq \overline{Y}_0$ and $\kappa(\mathcal{P}^1,Y_0)<\kappa(\mathcal{P}^2,\overline{Y}_0)$, then $\kappa(\mathcal{P}^2)> \kappa(\mathcal{P}^1)$.
\end{prop}
\begin{proof}
As $\mathcal{P}^2$ extends $\mathcal{P}^1$, for every $Y \in \mathcal{P}^1$, there exists a unique $\overline{Y} \in \mathcal{P}^2$ such that $Y \subseteq \overline{Y}$ with $\overline{Y}_1\neq \overline{Y}_2$ for all distinct $Y_1,Y_2 \in \mathcal{P}^1$.

We first show that $\alpha_{G_0,\mathcal{P}^1}(Y)\leq \alpha_{G_0,\mathcal{P}^2}(\overline{Y})$ for all $Y \in \mathcal{P}$.  Consider some $Y \in \mathcal{P}$.  As $Y \subseteq \overline{Y}$ and $\bigcup\mathcal{P}^1\setminus Y\subseteq \bigcup\mathcal{P}^2\setminus \overline{Y}$, we obtain that for every set $S \subseteq E(G_0)$ whose deletion separates all vertices in $Y$ from all vertices in $\bigcup\mathcal{P}^2\setminus Y$, we have that its deletion also separates all vertices in $\overline{Y}$ from all vertices in $\bigcup\mathcal{P}^1\setminus \overline{Y}$. This yields $\lambda_{G_0}(Y,\bigcup\mathcal{P}^1\setminus Y)\leq \lambda_{G_0}(\overline{Y},\bigcup\mathcal{P}^2\setminus \overline{Y})$. As $\mathcal{P}^1$ and $\mathcal{P}^2$ are states, we obtain that $\alpha_{G_0,\mathcal{P}^1}(Y)-\alpha_{G_0,\mathcal{P}^2}(\overline{Y})=\lambda_{G_0}(Y,\bigcup\mathcal{P}^1\setminus Y)-\lambda_{G_0}(\overline{Y},\bigcup\mathcal{P}^2\setminus \overline{Y})\leq 0$.
 
We now show that $\kappa(\mathcal{P}^1,Y)\leq \kappa(\mathcal{P}^2,\overline{Y})$ for all $Y \in \mathcal{P}^1$.  Consider some $Y \in \mathcal{P}$.   We distinguish some cases on the shape of $Y$. First suppose that $Y \in \mathcal{P}^1_{\textup{fat}}$. Then we have $\alpha_{G_0,\mathcal{P}^2}(\overline{Y})\geq \alpha_{G_0,\mathcal{P}^1}(Y)\geq \pi(t+1)+1$, so $\overline{Y}\in \mathcal{P}^2_{\textup{fat}}$. It follows that $\kappa(\mathcal{P}^2,\overline{Y})=2(\pi(t+1)+1)=\kappa(\mathcal{P}^1,Y)$. Next suppose that $Y\in \mathcal{P}^1_{\textup{f-n}}$, so there exists some $Y_0 \in \mathcal{P}^1_{\textup{fat}}, u \in Y$, and $v \in Y_0$ such that $uv \in E(G_0)$. By the same argument as before, we obtain that $\overline{Y}_0 \in \mathcal{P}^2_{\textup{fat}}$. It follows that $\overline{Y}\in \mathcal{P}^2_{\textup{fat}}\cup \mathcal{P}^2_{\textup{f-n}}$. If $\overline{Y}\in \mathcal{P}^2_{\textup{fat}}$, we obtain that $\kappa(\mathcal{P}^2,\overline{Y})= 2(\pi(t+1)+1)>\pi(t+1)+1+\alpha_{G_0,\mathcal{P}^1}(Y)=\kappa(\mathcal{P}^1,Y)$. If $\overline{Y}\in \mathcal{P}^2_{\textup{f-n}}$, we have $\kappa(\mathcal{P}^2,\overline{Y})-\kappa(\mathcal{P}^1,Y)=\alpha_{G_0,\mathcal{P}^2}(\overline{Y})-\alpha_{G_0,\mathcal{P}^1}(Y)\geq 0$. Finally suppose that $Y \in \mathcal{P}^1_{\textup{thin}}$. If $\overline{Y}\in \mathcal{P}^2_{\textup{fat}}\cup \mathcal{P}^2_{\textup{f-n}}$, we obtain $\kappa(\mathcal{P}^2,\overline{Y})\geq \pi(t+1)+1>\alpha_{G_0,\mathcal{P}^1}(Y)=\kappa(\mathcal{P}^1,Y)$. If $\overline{Y}\in \mathcal{P}^2_{\textup{thin}}$, we have $\kappa(\mathcal{P}^2,\overline{Y})-\kappa(\mathcal{P}^1,Y)=\alpha_{G_0,\mathcal{P}^2}(\overline{Y})-\alpha_{G_0,\mathcal{P}^1}(Y)\geq 0$.

We obtain $\kappa(\mathcal{P}^1)=\sum_{Y \in \mathcal{P}^1}\kappa(\mathcal{P}^1,Y)\leq \sum_{Y \in \mathcal{P}^1}\kappa(\mathcal{P}^2,\overline{Y})\leq \kappa(\mathcal{P}^2)$. Moreover, if equality holds, then $\kappa(\mathcal{P}^1,Y)=\kappa(\mathcal{P}^2,\overline{Y})$ holds for every $Y \in \mathcal{P}^1$.
\end{proof}

\subsection{Structural properties of minimum multicuts}\label{struktur}
This section is dedicated to proving some structural properties of minimum multicuts, which are crucially needed in our algorithm. In particular, we prove \cref{hauptred}. The first result is simple, but useful.
 \begin{prop}\label{nontriv}
 Let $(G,H)$ be a connected $\kplanar$ instance, $S$ a minimum multicut for $(G,H)$ and $Q$ a component of $G_0 \setminus S_0$. Then either $V(Q)\cap T \neq \emptyset$ or there exists some $e \in E_\pi$ that has exactly one endvertex in $V(Q)$.
 \end{prop}
 \begin{proof}
 Suppose otherwise. Observe that $Q$ is a component of $G \setminus S$. As $G$ is connected and $\delta_G(V(Q))\cap E_{\pi}=\emptyset$ by assumption, there exists some $e \in \delta_{G_0}(V(Q))$. Let $Q_0$ be the unique component of $G\setminus S$ containing the unique endvertex of $e$ that is not contained in $V(Q)$. Let $S'=S\setminus \{e\}$. Let $t_1,t_2 \in T$ with $t_1t_2\in E(H)$ and for $i \in[2]$, let $Q_i$ be the unique component of $G \setminus S$ containing $t_i$. Observe that $Q \notin \{Q_1,Q_2\}$ by the assumption that $V(Q)\cap T =\emptyset$. Further, as $S$ is a multicut for $(G,H)$, we have that $Q_1\neq Q_2$. Hence, by symmetry, we may suppose that $Q_0\neq Q_1$. We obtain that $\delta_G(V(Q_1))\subseteq S'$. It follows that $t_1$ and $t_2$ are in distinct components of $G \setminus S'$. This contradicts the minimality of $S$.
 \end{proof}
 
 The objective of the remainder of this section is to prove Lemma \ref{hauptred}.  We first need the following preparatory result.
\begin{lem}\label{2a1}
Let $(G,H)$ be a $\kplanar$ instance of $\mc$, let $S$ be a minimum multicut for $(G,H)$, let $i \in [2]$, and let $Q_1,Q_2$ be two components of $G_0\setminus S_0$  such that $(V(Q_1)\cup V(Q_2))\cap T \subseteq B_i \cup I$. Then $\delta_{G_0}(V(Q_1),V(Q_2))\leq \pi$.
\end{lem}
\begin{proof}
  Suppose otherwise and let $S'=(S \setminus \delta_{G_0}(V(Q_1),V(Q_2)))\cup E_\pi$. As $\delta_{G_0}(V(Q_1),V(Q_2))\subseteq S$  and $|E_\pi|=\pi$ by assumption, we have $|S'|<|S|$.

  We next show that $S'$ is a multicut for $(G,H)$. To this end, let $t_1,t_2 \in T$ with $t_1t_2 \in E(H)$. As $(B_1,B_2,I,X)$ is an extended biclique decomposition for $(G,H)$ and by symmetry, we may suppose that $t_1 \in B_{3-i}\cup X$. Let $Q$ be the component of $G_0\setminus S_0$ containing $t_1$. Observe that by construction, we have $Q\notin \{Q_1,Q_2\}$ and hence $\delta_G(V(Q))\subseteq \delta_{G_0}(V(Q))\cup E_\pi \subseteq S'$. It follows that there exists a collection of components of $G \setminus S'$ the union of whose vertex sets is $V(Q)$. Moreover, as $S$ is a multicut for $(G,H)$, we obtain that $t_2$ is not contained in $V(Q)$. It follows that $t_1$ and $t_2$ are in distinct components of $G \setminus S'$. 
  
 It follows that $S'$ is a multicut for $(G,H)$. As $|S'|<|S|$, this contradicts $S$ being a minimum multicut for $(G,H)$.
\end{proof}
We are now ready to prove Lemma \ref{hauptred}, which we restate here for convenience.
\hauptred*
\begin{proof}
Suppose otherwise, so $\textup{dist}_{G_0,S_0}(Y,X)\geq 2$.
By assumption, we in particular have $Y\cap X=\emptyset$. Hence, as $S$ is a multicut for $(G,H)$ and by symmetry, we may suppose that $Y \cap T \subseteq B_1 \cup I$. Next, let $\mathcal{Q}_1$ be the set of components of $G_0\setminus S_0$ that are linked to $Y$ by at least one edge and that contain at least one element of $B_1$ and let $F_1 \subseteq E(G_0)$ be the set of edges in $E(G_0)$ connecting $Y$ and $\bigcup_{Q \in \mathcal{Q}_1}V(Q)$. As $\textup{dist}_{G_0,S_0}(Y,X)\geq 2$, we have that $V(Q)\cap X=\emptyset$ for all $Q \in \mathcal{Q}_1$. It hence follows by Lemma \ref{2a1} that $d_{G_0}(Y,V(Q))\leq \pi$ for all $Q \in \mathcal{Q}_1$. As we clearly have $|\mathcal{Q}_1|\leq t$, we obtain $|F_1|\leq \pi t$. Now let $R\subseteq V(G_0)$  be a set such that $\delta_{G_0}(R)$ is a $(Y \cap (W \cup T),(W \cup T)\setminus Y)$-cut of size $\lambda(Y \cap (W \cup T),(W \cup T)\setminus Y)$ in $G_0$. We now define $S'=(S\setminus \delta_{G_0}(Y)) \cup (\delta_{G_0}(R) \cup F_1 \cup E_\pi)$. As $\delta_{G_0}(Y)\subseteq S$, by the definitions of $R$ and $F_1$ and as $Y$ is fat in $\mathcal{P}$, we have
\begin{align*}
|S'|&\leq |S|-d_{G_0}(Y)+d_{G_0}(R)+|F_1|+|E_\pi|\\
&\leq  |S|-(d_{G_0}(Y)-\lambda(Y \cap (W \cup T),(W \cup T)\setminus Y))+\pi t+\pi\\
&=|S|-\alpha_{G_0,\mathcal{P}}(Y)+\pi(t+1)\\
&<|S|.
\end{align*}

We next show that $S'$ is a multicut for $(G,H)$. Let $t_1,t_2 \in T$ with $t_1t_2 \in E(H)$. First suppose that $\{t_1,t_2\}\cap R \neq \emptyset$, say $t_1 \in R$. As $S$ is a multicut for $(G,H)$, it follows that $t_2 \notin Y$, so in particular $t_2 \notin R$. As $\delta_{G_0}(R)\cup E_{\pi}\subseteq S'$, we obtain that $t_1$ and $t_2$ are in distinct components of $G \setminus S'$.

Now suppose that $\{t_1,t_2\}\cap R = \emptyset$. As $(B_1,B_2,I,X)$ is an extended biclique decomposition of $(G,H)$ and by symmetry, we may suppose that $t_1 \in B_1 \cup X$. Let $Q$ be the component of $G_0\setminus S_0$ containing $t_1$.
 As $\textup{dist}_{G_0,S_0}(Y,X)\geq 2$ and by construction, we obtain that either $Q \in \mathcal{Q}_1$ or $E(G_0)$ does not contain an edge linking $V(Q)$ and $Y$. In either case, we obtain that $\delta_{G_0}(V(Q))\cup E_\pi\subseteq S'$. It follows that there exists a collection of components of $G \setminus S'$ the union of whose vertex sets is $V(Q)$. 
 As $S$ is a multicut for $(G,H)$, we obtain that $t_2 \notin V(Q)$. Hence $S'$ is a multicut for $(G,H)$. As $|S'|<|S|$, this contradicts $S$ being a minimum multicut for $(G,H)$.
\end{proof}

\subsection{Making states relevant}\label{relevant}

This section is dedicated to showing that, given an arbitrary state $\mathcal{P}$, we can efficiently compute another state $\mathcal{P}'$, which is relevant and maintains all important properties of $\mathcal{P}$. In other words, we prove \cref{makerelevant}. Most technical difficulties for this proof are contained in the following result, which shows that a single class of a state can be replaced by a larger one which is relevant in the arising state without losing any important property. 

\begin{lem}\label{makefat}
Let $(G,H)$ be a connected $\kplanar$ instance of $\mc$, and let $\mathcal{P}$ be a state. Further, let $Y_0 \in \mathcal{P}$, let $Y_0'$ be the relevant set for $(Y_0,\bigcup \mathcal{P}\setminus Y_0)$ in $G_0$ and let $\mathcal{P}'$ be obtained from $\mathcal{P}$ by replacing $Y_0$ by $Y_0'$. Then $\mathcal{P}'$ is a state with $\tau(\mathcal{P}')\leq \tau(\mathcal{P})$ and $\kappa(\mathcal{P}')\geq \kappa(\mathcal{P})$ and if $\mathcal{P}$ is maximum valid, so is $\mathcal{P}'$.
\end{lem}
\begin{proof}
It follows directly by construction that $\mathcal{P}'$ is a state.
We now show that $\tau(\mathcal{P}')\leq \tau(\mathcal{P})$ and $\kappa(\mathcal{P}')\geq \kappa(\mathcal{P})$. By \cref{relev1}, we have that every component of $Y_0'$ contains a component of $Y_0$. By the definition of $\mathcal{P}'$, this yields $\tau(\mathcal{P}')-\tau(\mathcal{P})=\tau(\{Y_0'\})-\tau(\{Y_0\})\leq 0$.

Next observe that $\kappa(\mathcal{P}')\geq \kappa(\mathcal{P})$ holds by \cref{kappaincrease}.

For the last part, we assume that $\mathcal{P}$ is maximum valid. We need to show that $\mathcal{P}'$ is maximum valid.
As $\mathcal{P}$ is maximum valid, there exists a minimum multicut $S$ for $(G,H)$ that respects $\mathcal{P}$. It follows in particular that for every $Y \in \mathcal{P}$, there exists some $\overline{Y}\in \mathcal{K}(G_0\setminus S_0)$ with $Y \subseteq \overline{Y}$. As $\mathcal{K}(G_0\setminus S_0)$ extends $\mathcal{P}$ and by the maximality of $\mathcal{P}$, we obtain that $\mathcal{K}(G_0\setminus S_0)=\{\overline{Y}:Y \in \mathcal{P}\}$. We define $\widetilde{Y}_0=\overline{Y}_0\cup Y_0'$ and $\widetilde{Y}=\overline{Y}\setminus Y_0'$ for all $Y \in \mathcal{P}\setminus \{Y_0\}$. Observe that, as $W \cup T \subseteq \bigcup \mathcal{P}$, as $\mathcal{K}(G_0\setminus S_0)$ extends $\mathcal{P}$, and by the definition of $Y_0'$, we have $(W \cup T)\cap Y=(W \cup T)\cap \overline{Y}=(W \cup T)\cap \widetilde{Y}$ for all $Y \in \mathcal{P}$. In particular, $\mathcal{P}'$ is a state. Next, we set $S'=(S \setminus \delta_{G_0}(\overline{Y}_0))\cup \delta_{G_0}(\widetilde{Y}_0)$. We will show that $S'$ is a minimum multicut for $(G,H)$ and $\mathcal{K}(G_0\setminus S'_0)=\{\widetilde{Y}:Y \in \mathcal{P}\}$. This directly implies that $S'$ respects $\mathcal{P}'$.

The following result is a first step in the direction of characterizing the components of $G_0 \setminus S_0'$. 
\begin{clm}\label{sevcomp}
For every $Y \in \mathcal{P}$, we have that $\widetilde{Y}$ is the union of a collection of classes of $\mathcal{K}(G_0\setminus S'_0)$.
\end{clm}
\begin{proof}
Suppose otherwise, so there exists an edge $e \in E(G_0)\setminus S_0'$ linking a vertex $v_1 \in \widetilde{Y}_1$ and a vertex $v_2 \in \widetilde{Y}_2$ for some $Y_1,Y_2 \in \mathcal{P}$. As $\delta_{G_0}(\widetilde{Y}_0)\subseteq S'_0$ by construction, we have $Y_0\notin \{Y_1,Y_2\}$. It follows by the definition of $\widetilde{Y}_1$ and $\widetilde{Y}_2$ that $v_1 \in \overline{Y_1}$ and $v_2 \in \overline{Y_2}$. This implies that $e \in S_0\setminus \delta_{G_0}(Y_0)$, so $e \in S'_0$, a contradiction.\claimqedhere
\end{proof}
We are now ready to show that $S'$ is a multicut indeed.
\begin{clm}\label{newmc}
$S'$ is a multicut for $(G,H)$.
\end{clm}
\begin{proof}
  We will show that if for some $t_1,t_2\in T \cup W$, there exists a $t_1t_2$-path in $G\setminus S'$, then there also exists a $t_1t_2$-path in $G\setminus S$. It suffices to prove the statement for the case that there exists a $t_1t_2$-path $P$ in $G\setminus S'$ none of whose interior vertices is contained in $T \cup W$.  In particular, if $P$ contains an edge of $E_\pi$, then $P$ is a single edge. As a first case, suppose that $P$ is a single edge of $E_\pi$. It follows that $P$ also exists in $G \setminus S$, so, in particular, there exists a $t_1t_2$-path in $G\setminus S$. Now suppose that $P$ is not a single edge in $E_\pi$ and hence that $E(P)\cap E_\pi=\emptyset$. We obtain that $t_1$ and $t_2$ are in the same component of $G_0\setminus S_0'$. It follows by Claim \ref{sevcomp} and as $\{\widetilde{Y}: Y \in \mathcal{P}\}$ is a partition of $V(G)$ that $\{t_1,t_2\}\subseteq \widetilde{Y}$ for some $Y \in \mathcal{P}$.  As $(W \cup T)\cap \overline{Y}=(W \cup T)\cap \widetilde{Y}$, this yields that $\{t_1,t_2\}\subseteq \overline{Y}$. As $\overline{Y}$ is a component of $G_0\setminus S_0$, we obtain in particular that $\overline{Y}$ is connected, so there exists a $t_1t_2$-path in $G\setminus S$.

 In all cases, there exists
a $t_1t_2$-path in $G\setminus S$. In particular, as $S$ is a multicut for $(G,H)$, we obtain that $G\setminus S'$ does not contain a $t_1t_2$-path for any $t_1,t_2 \in T$ with $t_1t_2 \in E(H)$.
\claimqedhere\end{proof}

We next show that $S'$ is of minimum size. 
\begin{clm}\label{minweight}
 $|S'|\leq |S|$.
\end{clm}
\begin{proof}
 As $Y_0'$ is the relevant set for $(Y_0,\bigcup \mathcal{P}\setminus Y_0)$ in $G_0$, we obtain that $d_{G_0}(\overline{Y}_0\cap Y_0')\geq d_{G_0}(Y_0')$. By the submodularity of the degree function, we obtain that $d_{G_0}(\widetilde{Y}_0)\leq d_{G_0}(\overline{Y}_0)+d_{G_0}(Y_0')-d_{G_0}(\overline{Y}_0\cap Y_0')\leq d_{G_0}(\overline{Y}_0)$.

By the minimality of $S$, together with $\delta_{G_0}(\overline{Y}_0)\subseteq S$, this yields 
\begin{align*}|S'|
&=|(S \setminus \delta_{G_0}(\overline{Y}_0))\cup \delta_{G_0}(\widetilde{Y}_0)|\\
&\leq|S \setminus \delta_{G_0}(\overline{Y}_0)|+d_{G_0}(\widetilde{Y}_0)\\
&=|S|- d_{G_0}(\overline{Y}_0)+d_{G_0}(\widetilde{Y}_0)\\
& \leq |S|.
\end{align*}
\claimqedhere\end{proof}
It follows from Claims \ref{newmc} and \ref{minweight} and the choice of $S$ that $S'$ is a minimum multicut for $(G,H)$.
With this result at hand, we are able to use the maximality of $\mathcal{P}$ to fully characterize the components of $G_0\setminus S'_0$. The next claim will directly imply that $\mathcal{K}(G_0\setminus S'_0)$ extends $\mathcal{P}'$ and hence finish the proof. 
\begin{clm}\label{comp}
$\mathcal{K}(G_0\setminus S'_0)=\{\widetilde{Y}:Y \in \mathcal{P}\}$.
\end{clm}
\begin{proof}
As $S'$ is a minimum multicut for $(G,H)$ by Claims \ref{newmc} and \ref{minweight}, we have that $\mathcal{K}(G_0\setminus S'_0)$ is a subpartition of $V(G)$ that is respected by a minimum multicut for $(G,H)$. If $G_0[\widetilde{Y}]$ is disconnected for some $Y \in \widetilde{Y}$ or $G_0\setminus S'_0$ contains a component which is distinct from $\widetilde{Y}$ for all $Y \in \mathcal{P}$, we obtain that $|\mathcal{K}(G_0\setminus S'_0)|>|\mathcal{P}|$, a contradiction to the maximality of $\mathcal{P}$.
\claimqedhere\end{proof}

\end{proof}
We are now ready to prove \cref{makerelevant}, which we restate here for convenience.
\makerelevant*

\begin{proof}

Let $Y_1,\ldots,Y_q$ be an arbitrary enumeration of $\mathcal{P}$. We now define a sequence $\mathcal{P}^0,\ldots,\mathcal{P}^q$  of states extending $\mathcal{P}$ such that $Y_{i+1},\dots,Y_q \in \mathcal{P}^{i}$ for $i \in [q]$. First, we set $\mathcal{P}^0=\mathcal{P}$. Next, for $i \in [q]$, let $Y_i'$ be the relevant set for $(Y_i,\bigcup \mathcal{P}^{i-1}\setminus Y_i)$. We obtain $\mathcal{P}^{i}$ from $\mathcal{P}^{i-1}$ by replacing $Y_i$ by $Y'_i$. We set $\mathcal{P}'=\mathcal{P}^q$. It follows by construction and Proposition \ref{relev3} that $Y'$ is the relevant set for $(Y',\bigcup \mathcal{P}'\setminus Y')$ for all $Y \in \mathcal{P}$.  Hence $\mathcal{P}'$ is a relevant state. Moreover, repeatedly applying \cref{makefat}, we obtain that $\tau(\mathcal{P}')\leq \tau(\mathcal{P})$, $\kappa(\mathcal{P}')\geq \kappa(\mathcal{P})$ and that, if $\mathcal{P}$ is maximum valid, then so is $\mathcal{P}'$. Further, we have $|\mathcal{P'}|=|\mathcal{P}|$ by construction. Finally, it follows from \cref{relev1} and that fact $q\leq 2\pi+t$ that $\mathcal{P}'$ can be computed in polynomial time. 
\end{proof}
\subsection{Dealing with incomplete states}\label{incompl}
In this section, we show how to replace an incomplete state by a bounded number of states for which $\kappa$ is larger. More precisely, we prove \cref{dealincomplete}, which we restate here for convenience.
\dealincomplete*
\begin{proof}
As $\mathcal{P}$ is incomplete, there exists a vertex $u_0 \in V(G)\setminus\bigcup \mathcal{P}$ that has a neighbor in some $Y_0\in \mathcal{P}_{\textup{thin}}$ in $G_0$. Now for every $Y \in \mathcal{P}$, let $Y'=Y\cup \{u_0\}$ and let $\mathcal{P}^{Y}$ be the subpartition of $V(G)$ in which $Y$ is replaced by $Y'$. It is easy to see that $\mathcal{P}^{Y}$ is a state for every $Y \in \mathcal{P}$. We now show that $(\mathcal{P}^{Y})_{Y \in \mathcal{P}}$ has the desired properties. First observe that $(\mathcal{P}^{Y})_{Y \in \mathcal{P}}$ can be computed in polynomial time by construction. Next, we have $|(\mathcal{P}^{Y})_{Y \in \mathcal{P}}|=|\mathcal{P}|\leq 2\pi+t$ as $\mathcal{P}$ is a state.

Let $Y \in \mathcal{P}$. Then $\mathcal{P}^{Y}$ is obtained from $\mathcal{P}$ by replacing $Y$ with $Y'$. As $Y'=Y\cup \{u_0\}$, we obtain that $G_0[Y']$ has at most one more component than $G_0[Y]$. This yields $\tau(\mathcal{P}^{Y})\leq  \tau(\mathcal{P})+1$.

We next show that $\kappa$ is increased for all newly created states, which is slightly more technical.
\begin{clm}
$\kappa(\mathcal{P}^{Y})>\kappa(\mathcal{P})$ for all $Y \in \mathcal{P}$.
\end{clm}
\begin{proof} Let $Y \in \mathcal{P}$. First observe that $\mathcal{P}^{Y}$ clearly extends $\mathcal{P}$.

 First consider the case that $Y \in \mathcal{P}_{\textup{f-n}}\cup \mathcal{P}_{\textup{thin}}$. We first show that $\kappa(\mathcal{P}^{Y},Y')>\kappa(\mathcal{P},Y)$. If $Y'\in  \mathcal{P}^{Y}_{\textup{fat}}$, we obtain that $\kappa(\mathcal{P}^{Y},Y')-\kappa(\mathcal{P},Y)\geq 2(\pi(t+1)+1)-(\pi(t+1)+1+\alpha_{G_0,\mathcal{P}}(Y))=\pi(t+1)+1-\alpha_{G_0,\mathcal{P}}(Y)>0$. We may hence suppose that $Y'\in \mathcal{P}^{Y}_{\textup{f-n}}\cup \mathcal{P}^{Y}_{\textup{thin}}$. It follows that $\kappa(\mathcal{P}^{Y},Y')-\kappa(\mathcal{P},Y)\geq \alpha_{G_0,\mathcal{P}^{Y}}(Y')-\alpha_{G_0,\mathcal{P}}(Y)$. As $\mathcal{P}$ is relevant, we have that $Y$ is the relevant set for $(Y,\bigcup \mathcal{P}\setminus Y)$, which yields $\lambda(Y',\bigcup \mathcal{P}\setminus Y)>\lambda(Y,\bigcup \mathcal{P}\setminus Y)$. By construction, we have  $\alpha_{G_0,\mathcal{P}^{Y}}(Y')-\alpha_{G_0,\mathcal{P}}(Y)=\lambda(Y',\bigcup \mathcal{P}^{Y}\setminus Y')-\lambda(Y,\bigcup \mathcal{P}\setminus Y)=\lambda(Y',\bigcup \mathcal{P}\setminus Y)-\lambda(Y,\bigcup \mathcal{P}\setminus Y)>0$. In either case, we obtain $\kappa(\mathcal{P}^{Y},Y')>\kappa(\mathcal{P},Y)$. It follows that $\kappa(\mathcal{P}^{Y})>\kappa(\mathcal{P})$ by \cref{kappaincrease}.

Now suppose that $Y\in \mathcal{P}_{\textup{fat}}$. We have  $\alpha_{G,\mathcal{P}^{Y}}(Y')=\lambda_{G_0}(Y',\bigcup \mathcal{P}^{Y}\setminus Y')-\lambda_{G_0}(Y'\cap (W \cup T), (W \cup T) \setminus Y')\geq \lambda_{G_0}(Y,\bigcup \mathcal{P}\setminus Y)-\lambda_{G_0}(Y\cap (W \cup T), (W \cup T) \setminus Y)=\alpha_{G,\mathcal{P}}(Y)$. It follows that $Y' \in \mathcal{P}^{Y}_{\textup{fat}}$. As $u_0$ has a neighbor in $Y_0$ in $G_0$, it follows by construction and definition that $Y_0 \in \mathcal{P}^{Y}_{\textup{fat}}\cup \mathcal{P}^{Y}_{\textup{f-n}}$. This yields $\kappa(\mathcal{P}^{Y},Y_0)-\kappa(\mathcal{P},Y_0)\geq \pi(t+1)+1-\alpha_{G_0,\mathcal{P}}(Y_0)>0$. Again, it follows that $\kappa(\mathcal{P}^{Y})>\kappa(\mathcal{P})$ by \cref{kappaincrease}.
\claimqedhere\end{proof}

For the rest of the proof, suppose that $\mathcal{P}$ is maximum valid, so there exists a minimum multicut $S$ for $(G,H)$ that respects $\mathcal{P}$. It follows that for every $Y \in \mathcal{P}$, there exists a unique $\overline{Y}\in \mathcal{K}(G_0\setminus S_0)$ such that $Y \subseteq \overline{Y}$ with $\overline{Y}_1\neq \overline{Y}_2$ for all distinct $Y_1,Y_2 \in \mathcal{P}$. By the maximality of $\mathcal{P}$, we obtain that $\mathcal{K}(G_0\setminus S_0)=(\overline{Y}:Y \in \mathcal{P})$. In particular, there exists a unique $Y_1 \in \mathcal{P}$ such that $u_0 \in \overline{Y}_1$. It follows by definition that $S$ respects $\mathcal{P}^{Y}_1$, so $\mathcal{P}^{Y_1}$ is valid. As $|\mathcal{P}^{Y_1}|=|\mathcal{P}|$ and $\mathcal{P}$ is maximum valid, so is $\mathcal{P}^{Y_1}$. 
\end{proof}

\subsection{Dealing with complete states}\label{compl}
In this section, we show how to continue with our algorithm once a complete maximum valid state $\mathcal{P}$ is attained. More precisely, we prove \cref{phigross} and \cref{compphiklein}. The case that the state in consideration contains a thin class is dealt with in Lemma \ref{phigross}. The other case is when $\mathcal{P}_{\textup{thin}}$ is empty. In this case, we show that a treewidth bound on the multicut dual can be proved and hence a minimum multicut can be directly computed. We need to combine some topological arguments with the machinery developped in \cite{focke_et_al:LIPIcs.SoCG.2024.57}. The result is proved in Lemma \ref{compphiklein}.

We start by dealing with the first case. More precisely, we prove \cref{phigross}, which we restate here for convenience.
\phigross* 
\begin{proof}
We first choose some thin $Y_0 \in \mathcal{P}_{\textup{thin}}$ and define $S'=\delta_{G_0}(Y_0)$. Clearly, $S'$ can be computed in polynomial time. 

Next, as $\mathcal{P}$ is valid, there exists a minimum multicut $S$ for $(G,H)$ that respects $\mathcal{P}$. For $Y \in \mathcal{P}$, as $\mathcal{K}(G_0\setminus S_0)$ extends $\mathcal{P}$, there exists a unique $\overline{Y}\in \mathcal{K}(G_0\setminus S_0)$ with $Y \subseteq \overline{Y}$. Now consider an edge $v_0v_1 \in S'$ with $v_0 \in Y_0$. As $\mathcal{P}$ is complete, there exists some $Y_1 \in \mathcal{P}$ such that $v_1 \in Y_1$. We obtain that $v_0 \in \overline{Y}_0$ and $v_1 \in \overline{Y}_1$. It follows that $S' \subseteq S$ and that $\overline{Y}_0=Y_0$. For the remainder of the proof, we need a small case distinction. By Proposition \ref{nontriv} and the fact that $Y_0\in \mathcal{K}(G_0\setminus S_0)$, one of the following two cases occurs. Further, we can clearly observe in polynomial time which of the two cases occurs.
\begin{Casex}
There exists an edge $e \in E_\pi$ with exactly one endvertex in $Y_0$.
\end{Casex}
\begin{proof}
We now define $G'=G \setminus S'$ and $H'=H$. Clearly, $(S',(G',H'))$ can be computed in polynomial time and is an extended subinstance of $(G,H)$. Further, as $G'\setminus E_\pi$ is planar and $G'[Y_0]$ is a component of $G'\setminus E_\pi$, we obtain by Lemma \ref{conn} that $G'\setminus (E_\pi \setminus \{e\})$ is planar, so $\pi(G')<\pi(G)$. Now let $S''$ be a minimum multicut for $(G',H')$. Then, as $S$ is a minimum multicut for $(G,H)$ and $S' \subseteq S$, we have that $S' \cup S''$ is a minimum multicut for $(G,H)$ by Proposition \ref{cutequiv}. Hence $(S',(G',H'))$ is optimumbound.
\claimqedhere\end{proof}
\begin{Casex}
$Y_0 \cap W=V(E')$ for some $E' \subseteq E_\pi$ and $Y \cap T \neq \emptyset$.
\end{Casex}
\begin{proof}
We now define $G'=G- Y_0$ and $H'=H- Y_0$. Clearly, $(S',(G',H'))$ can be computed in polynomial time and is an extended subinstance of $(G,H)$. Further, $|V(H')|<|V(H)|$ holds by assumption.  

Now let $S''$ be a minimum multicut for $(G',H')$. We first show that $S' \cup S''$ is a multicut for $(G,H)$. Let $t_1,t_2 \in V(H)$ with $t_1t_2 \in E(H)$. If $\{t_1,t_2\}\subseteq Y_0$, we obtain an immediate contradiction as $Y_0$ is a component of $G \setminus S$ and $S$ is a multicut for $(G,H)$. We may hence suppose that $\{t_1,t_2\} \setminus Y_0 \neq \emptyset$. Suppose for the sake of a contradiction that there exists a $t_1t_2$-path $P$ in $G \setminus (S' \cup S'')$. If $V(P)\cap Y_0=\emptyset$, we obtain that $P$ also exists in $G'\setminus S''$, a contradiction to $S''$ being a multicut for $(G',H')$. Otherwise, we obtain that $E(P)$ contains an edge $e$ with exactly one endvertex in $Y_0$. If $e \in E(G_0)$, we obtain a contradiction as $\delta_{G_0}(Y_0)\subseteq S' \cup S''$. If $e \in E_\pi$, we obtain a contradiction to the assumption that $Y_0 \cap W=V(E')$ for some $E' \subseteq E_\pi$.

Observe that, as $G'\setminus (S \setminus S')$ is a subgraph of $G \setminus S$, we have that $(S \setminus S')\cap E(G')$ is a multicut for $(G',H')$. By the choice of $S''$, we obtain that $|S' \cup S''|=|S'|+|S''|\leq |S'|+|S \setminus S'|=|S|$. Hence $S' \cup S''$ is a minimum multicut for $(G,H)$, so $(S',(G',H'))$ is optimumbound.
\claimqedhere\end{proof}
\end{proof}

 We now consider the case that $\mathcal{P}$ does not contain a thin set. Our objective is to prove \cref{compphiklein}. We want to make use of the fact that the problem can be reduced to computing a multicut of a planar instance of \mc{} such that every minimum multicut dual for this instance is of small treewidth. The key idea is that for any minimum multicut $S$ for $(G,H)$, we have that $\textup{dist}_{G_0,S_0}(v,X)$ is small for all $v \in V(G)$. The first result shows that this property indeed guarantees that all minimum multicut duals are of small treewidth.
\begin{lem}\label{disttw}
Let $(G,H)$ be a planar instance of \mc, let $(B_1,B_2,I,X)$ an extended biclique decomposition of $H$ with $|X|=\mu$, let $C$ be an inclusionwise minimal multicut dual for $(G,H)$, let $S=e_G(C)$ and suppose that $\textup{dist}_{G,S}(v,X)\leq 2$ for every $v \in V(G)$. Then $\tw(C)=O(\sqrt{\mu})$.
\end{lem}
\begin{proof}
We now construct a graph $C'$ from $C$ by adding a new vertex $z_f$ for every face $f$ of $C$ and linking it to all vertices in $V(C)$ incident to $f$. Observe that $C'$ is planar. Next, let $Z_X\subseteq V(C')$ be the set containing the vertex $z_{f}$ for every face $f$ of $C$ containing some $x \in X$.

\begin{clm}
$Z_X$ is a 5-dominating set of $C'$.
\end{clm}    
\begin{proof}
It suffices to prove that for every face $f$ of $C$, there exists a path of length at most 4 in $C'$ from $z_f$ to $Z_X$. Let $f$ be a face of $C$. If $f$ did not contain a vertex of $V(G)$, then another multicut dual could be obtained from $C$ by removing an edge incident to $f$.  Hence, by the minimality of $C$, there exists some $u \in V(G)$ that is contained in $f$. As $\textup{dist}_{G,S}(v,X)\leq 2$ for every $v \in V(G)$, we obtain that there exists a $ux$-path in $G$ for some $x \in X$ that contains at most two edges of $S$. Let $Q_1,\ldots,Q_k$ be the components of $G \setminus S$ traversed by this path and observe that $k \leq 3$. Further, for $i \in [k]$, let $f_i$ be the unique face containing $V(Q_i)$. Now, as $x \in V(Q_k)$, we have $z_{f_k}\in Z_X$. For $i \in [k-1]$, observe that, as $Q_i$ and $Q_{i+1}$ are linked by an edge of $S$ and $S=e_G(C)$, we have that $f_i$ and $f_{i+1}$ are adjacent faces of $C$. Let $v_i \in V(C)$ be a vertex incident to both $f_i$ and $f_{i+1}$. We obtain that $z_{f_1}v_1z_{f_2}\ldots z_{f_k}$ is a path of length at most 4 in $C'$ from $z_f$ to $Z_X$.
\claimqedhere\end{proof}
Finally, we clearly have $|Z_X|\leq |X|=\mu$. It hence follows from Proposition \ref{domplanar} that $tw(C)\leq tw(C')=O(\sqrt{\mu})$.
\end{proof}
 With the combinatorial bound of \cref{disttw} and the algorithm of \cref{algogenext}, we can show  that the above described distance property for minimum multicuts guarantees an efficient algorithm for planar instances.
\begin{lem}\label{compdistsmall}
Let $(G,H)$ be an edge-weighted planar instance of \mc\ such that for every minimum-weight multicut $S$ for $(G,H)$ and every $v \in V(G)$, we have $\textup{dist}_{G,S}(v,X)\leq 2$. Then an optimum solution for $(G,H)$ can be computed in $f(t)n^{O(\sqrt{\mu})}$.
\end{lem}
\begin{proof}
Let $C$ be an inclusionwise minimal, minimum weigh-multicut dual for $(G,H)$ and let $S=e_G(C)$. By \cref{dualcut}, we have that $S$ is a minimum-weight multicut for $(G,H)$. It hence follows by assumption that $\textup{dist}_{G,S}(v,X)\leq 2$ holds for all $v \in V(G)$. It now follows from \cref{disttw} that $\tw(C)=O(\sqrt{\mu})$. The statement now follows from \cref{algogenext} for $F^*=\emptyset$.
\end{proof}
We are now ready to prove \cref{compphiklein}, which is the second main result of this section. It shows how to obtain a minimum multicut when a complete maximum valid state with no thin class is available.
\compphiklein*

\begin{proof}
We start by defining an edge-weighted planar instance $(G_0',H')$ of \mc. We first define $G'_0$. Namely, for all $Y \in \mathcal{P}$ and every $e \in E(G_0[Y])$, we set the weight of $e$ to $\infty$. All remaining edges of $G'_0$ are of weight $1$. We now define $H'$. First, for every $Y \in \mathcal{P}$ and every component $Q$ of $G_0[Y]$, we choose a vertex $v_Q\in V(Q)$ and let $V(H')$ contain $v_Q$. Next, for any distinct $Y_1,Y_2 \in \mathcal{P}$, any component $Q_1$ of $G_0[Y_1]$ and any component $Q_2$ of $G_0[Y_2]$, we let $V(H')$ contain an edge linking $v_{Q_1}$ and $v_{Q_2}$ (thus $H'$ is a complete multipartite graph). 

In the following, we will show that a minimum multicut for $(G,H)$ can easily be obtained from a minimum multicut for $(G_0',H')$ and, moreover, that a minimum multicut for $(G_0',H')$ can be efficiently computed.
We first show that a minimum multicut for $(G,H)$ can be obtained from a minimum multicut for $(G_0',H')$.
\begin{clm}\label{ismc}Let $S$ be a minimum multicut for $(G,H)$ that respects $\mathcal{P}$, let $S_\pi=S \cap E_\pi$, and let $S'_0$ be a minimum multicut for $(G'_0,H')$ and $S'=S_0' \cup S_\pi$. Then $S'$ is a multicut for $(G,H)$.
\end{clm}
\begin{proof}As $S$ is a multicut for $(G,H)$ and by construction, we obtain that $S_0$ is a multicut for $(G'_0,H')$ of finite weight. By the minimality of $S_0'$, it follows that $S_0'$ is of finite weight. 
 We will show that for any $t_1,t_2 \in W \cup T$, if there exists a $t_1t_2$-path in $G\setminus S'$, then there also exists a $t_1t_2$-path in $G\setminus S$. It suffices to consider the case that there exists a $t_1t_2$-path $P$ in $G\setminus S'$ none of whose interior vertices is contained in $W \cup T$. If $E(P)\cap E_\pi \neq \emptyset$, we obtain that $E(P)$ consists of a single edge in $E_\pi \setminus S_\pi$. Then $P$ also exists in $G\setminus S$. Otherwise, we obtain that $P$ is fully contained in $G_0\setminus S'_0$. For $i \in [2]$, let $Y_i$ be the unique partition class of $\mathcal{P}$ containing $t_i$ and let $Q_i$ be the component of $G_0[Y_i]$ containing $t_i$. As $S'_0$ is of finite weight and by construction, we obtain that $v_{Q_i}$ is in the same component of $G_0'\setminus S_0'$ as $t_i$. As $P$ is fully contained in $G_0\setminus S'_0$, we obtain that $v_{Q_1}$ and $v_{Q_2}$ are in the same component of $G_0'\setminus S_0'$. If $Y_1\neq Y_2$, we obtain that $H'$ contains an edge linking $v_{Q_1}$ and $v_{Q_2}$, contradicting $S_0'$ being a multicut for $(G_0',H')$. It follows that $Y_1=Y_2$. As $S$ respects $\mathcal{P}$, we obtain that $t_1$ and $t_2$ are contained in the same component of $G_0 \setminus S_0$. Clearly, it follows that there exists a $t_1t_2$-path in $G\setminus S$. This yields that $S'$ is a multicut for $(G,H)$.
\claimqedhere\end{proof}

The next already implies that computing a minimum multicut for $(G_0',H')$ is sufficient.
\begin{clm}\label{ismmc}Let $S$ be a minimum multicut for $(G,H)$ that respects $\mathcal{P}$, let $S_\pi=S \cap E_\pi$, let $S'_0$ be a minimum multicut for $(G'_0,H')$ and $S'=S_0' \cup S_\pi$. Then $S'$ is a minimum multicut for $(G,H)$.
\end{clm}
\begin{proof}
It follows from Claim \ref{ismc} that $S'$ is a multicut for $(G,H)$. Next, observe that by the definition of $H'$, we have that $S_0$ is a multicut for $(G_0',H')$. Further, as $S$ respects $\mathcal{P}$ and the minimality of $S$, we have that $w(e)=1$ for all $e \in S_0$, where $w$ is the weight function associated to $G_0'$. Hence, by the definition of $w$ and the minimality of $S_0'$, we obtain $|S'|=|S_0'|+|S_\pi|\leq w(S_0')+|S_\pi|\leq w(S_0)+|S_\pi|=|S_0|+|S_\pi|=|S|$. Hence, as $S$ is a minimum multicut for $(G,H)$, so is  $S'$.
\claimqedhere\end{proof}
In order to show that a minimum multicut for $(G_0',H')$ can be computed efficiently, we need the following strengthening of Claim \ref{ismmc}.
\begin{clm}\label{ismmcpr}Let $S$ be a minimum multicut for $(G,H)$ that respects $\mathcal{P}$, let $S_\pi=S \cap E_\pi$, let $S'_0$ be a minimum multicut for $(G'_0,H')$, and let $S'=S_0' \cup S_\pi$. Then $S'$ is a minimum multicut for $(G,H)$ that respects $\mathcal{P}$.
\end{clm}
\begin{proof}
By Claim \ref{ismmc}, we only need to show that $S'$ that respects $\mathcal{P}$.
Let $Y_1,Y_2 \in \mathcal{P}$ with $Y_1 \neq Y_2$ and let $y_1\in Y_1$ and $y_2\in Y_2$. We start by showing that $y_1$ and $y_2$ are in distinct components of $G'_0\setminus S'_0$. To this end, for $i \in [2]$, let $Q_i$ be the component of $G_0[Y_i]$ containing $y_i$. Observe that, as $S_0'$ is finite and by construction, we have that $y_i$ and $v_{Q_i}$ are in the same connected component of $G'_0\setminus S_0'$. Next, by construction, we have $v_{Q_1}v_{Q_2}\in E(H')$. As $S'_0$ is a multicut for $(G_0',H')$, we obtain that $v_{Q_1}$ and $v_{Q_2}$ are in distinct components of $G_0'\setminus S_0'$. It follows that $y_1$ and $y_2$ are in distinct components of $G_0'\setminus S_0'$.

We hence obtain that for every $Z \in \mathcal{K}(G_0\setminus S'_0)$, there exists at most one $Y \in \mathcal{P}$ with $Z \cap Y \neq \emptyset$. Hence, if $\mathcal{K}(G_0\setminus S'_0)$ does not extend $\mathcal{P}$, we obtain that $|\mathcal{K}(G_0\setminus S'_0)|>|\mathcal{P}|$. This contradicts the maximality of $\mathcal{P}$. It follows that $\mathcal{K}(G_0\setminus S'_0)$ extends $\mathcal{P}$, so $S'$ respects $\mathcal{P}$.
\claimqedhere\end{proof}

We are now ready to show that a minimum multicut for $(G'_0,H')$ can be computed efficiently.
\begin{clm}\label{compute}We can compute a minimum multicut for $(G'_0,H')$ in $f(|V(H')|)n^{O(\sqrt{\mu})}$.
\end{clm}
\begin{proof}
Let $C$ be an inclusionwise minimal minimum weight multicut dual for $(G'_0,H')$ and let $S'_0=e_{G'_0}(C)$. By Lemma \ref{dualcut}, we have that $S'_0$ is a minimum weight multicut for $(G'_0,H')$. Now, as $\mathcal{P}$ is maximum valid, there exists a minimum multicut $S$ for $(G,H)$ that respects $\mathcal{P}$. Let $S_\pi=S \cap E_\pi$ and let $S'=S'_0 \cup S_\pi$. It follows from Claim \ref{ismmcpr} that $S'$ is a minimum multicut for $(G,H)$ that respects $\mathcal{P}$, so for every $Y \in \mathcal{P}$, there exists a unique $\overline{Y}\in \mathcal{K}(G_0\setminus S'_0)$ with $Y \subseteq \overline{Y}$.  It follows from \cref{hauptred} that $\textup{dist}_{G_0\setminus S'_0}(v,X)\leq 1$ for all $v \in \overline{Y}$ and all $Y \in \mathcal{P}_{\textup{fat}}$. Next for every $Y \in \mathcal{P}_{\textup{f-n}}$, by definition, there exists some $Y'\in \mathcal{P}_{\textup{fat}}$ such that $E(G_0)$ contains an edge linking $Y$ and $Y'$. Observe that this edge also links $\overline{Y}$ and $\overline{Y'}$. It follows that $\textup{dist}_{G_0\setminus S'_0}(v,X)\leq 2$ for all $v \in \overline{Y}$ and all $Y \in \mathcal{P}_{\textup{f-n}}$. As $G$ is connected and by Lemma \ref{conn}, we have that $G_0$ is connected. Hence, by the minimality of $S_0'$, we obtain that $\bigcup_{Y \in \mathcal{P}}\overline{Y}=V(G)$. As $\mathcal{P}_{\textup{thin}}=\emptyset$ by assumption, it follows that $\textup{dist}_{G_0,S'_0}(v,X)\leq 2$ holds for all $v \in V(G)$.

It now follows from \cref{compdistsmall} that a minimum multicut for $(G_0',H')$ can be computed in $f(|V(H')|)n^{O(\sqrt{\mu})}$.
\claimqedhere\end{proof}

We are now ready to conclude the main proof. We first compute a minimum multicut $S_0$ for $(G_0',H')$.  Now for every $S_\pi \subseteq E_\pi$, we check whether $S_0\cup S_\pi$ is a multicut for $(G,H)$ and we output the minimum multicut found during this procedure. As $\mathcal{P}$ is maximum valid, there is a minimum multicut for $(G,H)$ that respects $\mathcal{P}$. It hence follows from Claim \ref{ismmc} that our algorithm outputs a minimum multicut for $(G,H)$.

For the running time of the algorithm, first observe that $(G_0',H')$ can clearly be computed from $(G,H)$ and $\mathcal{P}$ in polynomial time. Next observe that by construction, we have that $|V(H')|\leq \tau(\mathcal{P})$. It hence follows by Claim \ref{compute} that we can compute a minimum multicut $S_0$ for $(G_0',H')$ in $f(\tau(\mathcal{P}))n^{O(\sqrt{\mu})}$.  Moreover, as there are only $2^\pi$ choices for $S_\pi$ and we can check in polynomial time whether $S_0 \cup S_\pi$ is a multicut for $(G,H)$, we have that the total running time of our algorithm is $f(\tau(\mathcal{P}),\pi)n^{O(\sqrt{\mu})}$.
\end{proof}

\subsection{Reducing to smaller instances}\label{reduce}
In this section, we combine the previously established subroutines to reduce a given $\kplanar$ instance of $\mc$. 
More concretely, the rest of Section \ref{reduce} is dedicated to proving Lemma \ref{creins}. The lemmas stated earlier make it more or less straightforward how this can be done, but we describe the process in detail for completeness.

In Lemma \ref{initialize}, we show how an initial collection of states one of which is maximum valid can efficiently be made available. Next, in Lemma \ref{crzwei}, we show how to prove the result once a maximum valid state is available. Afterwards, Lemma \ref{creins} follows easily.
\begin{lem}\label{initialize}
Let $(G,H)$ be a connected $\kplanar$ instance of $\mc$. Then, in $f(\pi,t)n^{O(1)}$, we can compute a collection $\mathbb{P}$ of $f(\pi,t)$ states such that at least one of them is maximum valid and $\tau(\mathcal{P})\leq 2 \pi+t$ holds for all $\mathcal{P}\in \mathbb{P}$.
\end{lem}
\begin{proof}

We let $\mathbb{P}$ be the set of partitions of $W \cup T$. Clearly, $\mathbb{P}$ can be computed in $f(\pi,t)n^{O(1)}$ and $\tau(\mathcal{P})\leq 2 \pi+t$ holds for all $\mathcal{P}\in \mathbb{P}$. By definition, every $\mathcal{P}\in \mathbb{P}$ is a state. It hence suffices to prove that one of them is maximum valid. Let $S$ be a minimum multicut for $(G,H)$ that maximizes $|\mathcal{K}(G_0\setminus S_0)|$ and let $\mathcal{P}$ be the partition of $W \cup T$ which is defined by $\mathcal{P}=\{Y \cap (W \cup T): Y \in \mathcal{K}(G_0\setminus S_0)\}$. Observe that by Proposition \ref{nontriv}, we have that $Y \cap (W \cup T)\neq \emptyset$ for $Y \in \mathcal{K}(G_0\setminus S_0)$ and hence $\mathcal{P}\in \mathbb{P}$. Then, clearly $\mathcal{K}(G_0\setminus S_0)$ extends $\mathcal{P}$ and hence $S$ respects $\mathcal{P}$. It follows that $\mathcal{P}$ is valid. Further, we have $|\mathcal{P}|=|\mathcal{K}(G_0\setminus S_0)|$, so $\mathcal{P}$ is maximum.
\end{proof}
The most difficult part of the recursive procedure for computing our minimum multicut is encoded in the following algorithm. It shows how the desired output can be obtained once a maximum valid state is available. The algorithm requires the construction of a search tree, which is slightly technical. However, the idea of the algorithm is conceptually not difficult now that all necessary subroutines are available.

\begin{lem}\label{crzwei}
Let $(G,H)$ be a connected $\kplanar$ instance of $\mc$ and let a maximum valid state $\mathcal{P}$ with $\tau(\mathcal{P})\leq 2 \pi+t$ be given. Then, in $f(\pi,t)n^{O(\sqrt{\mu})}$, we can run an algorithm that returns a set $S$  and a collection of extended subinstances $(S_i,(G_i,H_i))_{i \in [k]}$ of $(G,H)$  for some $k=f(\pi,t)$ such that for all $i \in [k]$, we have that $\pi(G_i)+|V(H_i)|<\pi(G)+|V(H)|$ holds and either $S$ is a minimum multicut for $(G,H)$ or there exists some $i \in [k]$ such that $(S_i,(G_i,H_i))$ is optimumbound.
\end{lem}
\begin{proof}
We create a rooted search tree each of whose nodes is associated to a state of $(G,H)$. First, we initialize the tree by its root which is associated to $\mathcal{P}$. 

Now consider a leaf of the current tree that has not yet been considered after its creation and is associated to an incomplete state $\mathcal{P}^{0}$.

 First suppose that $\mathcal{P}^{0}$ is not relevant. In that case, by Lemma \ref{makerelevant}, in polynomial time, we can compute a relevant state $\mathcal{P}^1$ with $|\mathcal{P}^1|=|\mathcal{P}^{0}|, \tau(\mathcal{P}^1)\leq \tau(\mathcal{P}^{0}), \kappa(\mathcal{P}^1)\geq \kappa(\mathcal{P}^{0})$ and such that $\mathcal{P}^1$ is maximum valid if $\mathcal{P}^{0}$ is. We create a single child of the leaf in consideration and associate it to $\mathcal{P}^1$. 

Now suppose that $\mathcal{P}^{0}$ is relevant. By Lemma \ref{dealincomplete}, there exists an algorithm that runs in polynomial time which we can apply to $\mathcal{P}^{0}$ and that returns a collection of $k$ states $(\mathcal{P}^{i})_{i \in [k]}$ for some $k \leq 2\pi+t$ such that $\kappa(\mathcal{P}^{i})>\kappa(\mathcal{P}^{0})$ and $\tau(\mathcal{P}^{i})\leq \tau(\mathcal{P}^{0})+1$ holds for $i \in [k]$ and such that, if $\mathcal{P}^{0}$ is maximum valid, then  $\mathcal{P}^{i}$ is maximum valid for some $i \in [k]$. If the algorithm returns a collection of $k$ states $(\mathcal{P}^{i})_{i \in [k]}$ for some $k \leq 2\pi+t$ such that $\kappa(\mathcal{P}^{i})>\kappa(\mathcal{P}^{0})$ and $\tau(\mathcal{P}^{i})\leq \tau(\mathcal{P}^{0})+1$ for $i \in [k]$, then we add $k$ nodes to the rooted tree which are children of the node in consideration and associate each of them to the state $\mathcal{P}^{i}$ for some $i \in [k]$ in a  way that $\mathcal{P}^{i}$ is associated to one of these nodes for $i \in [k]$.

We do this as long as there exists a leaf of the current tree which has not been considered after its creation and is associated to an incomplete state. Let $U$ be the final tree at the end of this procedure. Observe that $\kappa(\mathcal{P})\geq 0$ by definition and, for every leaf of $U$ that is associated to a state $\mathcal{P}'$, we have  $\kappa(\mathcal{P}')=f(\pi,t)$ as $\mathcal{P}'$ is a state and by the definition of $\kappa$. Further, for every node that is associated to a state $\mathcal{P}^{0}$ and every child of this node that is associated to a state $\mathcal{P}^1$, we have $\kappa(\mathcal{P}^1)\geq \kappa(\mathcal{P}^{0})$ by assumption. Moreover, if $\mathcal{P}^{0}$ is relevant, than the inequality is strict. As all children of nodes associated to an irrelevant state are relevant, it follows that the depth of $U$ is $f(\pi,t)$. In particular, $U$ is well-defined. Further observe that for every node that is associated to a state $\mathcal{P}^{0}$ and every child of this node that is associated to a state $\mathcal{P}^1$, we have $\tau(\mathcal{P}^1)\leq\tau(\mathcal{P}^{0})+1$ by assumption. As the depth of $U$ is $f(\pi,t)$ and $\tau(\mathcal{P})\leq 2\pi+t$, we obtain that $\tau(\mathcal{P}^{0})=f(\pi,t)$ for every state $\mathcal{P}^{0}$ that is associated to a leaf of $U$. Further observe that every leaf of $U$ either corresponds to a state that is not maximum valid or to a complete state. Next observe that, as $\mathcal{P}$ is maximum valid and every node of $U$ that is associated to a maximum valid state has a child that is associated to a maximum valid state, we obtain that there exists a leaf of $U$ that is associated to a maximum valid complete state. 

For the next part of the algorithm, we maintain a multicut $S$ for $(G,H)$, starting with $E(G)$. Now consider a leaf of $U$ which is associated to a complete state  $\mathcal{P}^{0}$. Recall that $\tau(\mathcal{P}^{0})=f(\pi,t)$. By Lemma \ref{compphiklein}, if $\mathcal{P}^{0}$ is maximum valid and  $\mathcal{P}^{0}_{\textup{thin}}=\emptyset$, there exists an algorithm we can apply to $(G,H)$ that runs in $f(\pi,t)n^{O(\sqrt{\mu})}$ and that returns a minimum multicut to $(G,H)$. We apply this algorithm to $(G,H)$  and if it returns a multicut $S'$ for $(G,H)$ with $|S'|<|S|$, then we replace $S$ by $S'$. We do this for all leaves of $U$ which are associated to a complete state. The assignment of $S$ after the last of the iterations is the first part of our output.

For the last part of our algorithm, again consider a leaf of $U$ which is associated to a complete state  $\mathcal{P}^{0}$. It follows from Lemma \ref{phigross} that there exists an algorithm we can apply to $(G,H)$ that runs in polynomial time and that, in case that $\mathcal{P}^{0}$ is valid and contains a thin partition class, returns an optimumbound extended subinstance of $(G,H)$. If the algorithm returns an extended subinstance of $(G,H)$, we add it to the collection of extended subinstances we return. We do this for all leaves of $U$ which are associated to a complete state.

For the correctness of the algorithm, recall that there exists a leaf of $U$ that is associated to a maximum valid state $\mathcal{P}^{0}$. If $\mathcal{P}^{0}_{\textup{thin}}=\emptyset$, then it follows from the assumption on the algorithm and by construction that the returned set $S$ is a minimum multicut for $(G,H)$. If $\mathcal{P}^{0}_{\textup{thin}}\neq\emptyset$, then it follows from the assumption on the algorithm of Lemma \ref{phigross} and by construction that the algorithm returns a collection of extended subinstances of $(G,H)$ containing an optimumbound one. Let $k$ be the number of leaves of $U$. Observe that the number of extended subinstances of $(G,H)$ that are returned is clearly bounded by $k$. Further observe that, as every node of $U$ has $f(\pi,t)$ children and the depth of $U$ is $f(\pi,t)$, we have that $k=f(\pi,t)$. Hence the size of the output is appropriate.

For the running time, observe that, as every node of $U$ has $f(\pi,t)$ children and the depth of $U$ is $f(\pi,t)$, we obtain that the number of nodes of $U$ is $f(\pi,t)$. It hence suffices to prove that the running time of the operations associated to every single node of $U$ is bounded by $f(\pi,t)n^{O(\sqrt{\mu})}$. For every node that is associated to an incomplete state, observe that by assumption, the algorithm for potentially computing the states associated to its children runs in polynomial time. Further, for all possible outputs, we can clearly check in polynomial time if it is a collection of states of the appropriate size. Hence the computation time associated to this node is polynomial. Now consider a node associated to a complete state. By assumption, the algorithm used for computing a potential multicut for $(G,H)$ associated to this node runs in $f(\pi,t)n^{O(\sqrt{\mu})}$. Further, we can check whether the output is a multicut in polynomial time by  Proposition \ref{checkmc}. For the second part, the algorithm potentially computing an extended subinstance of $(G,H)$ runs in polynomial time by assumption. Finally, we can clearly check whether the output is an instance with the described properties in polynomial time. Hence the computation time associate to this node is also polynomial. This finishes the proof.
\end{proof}
We are now ready to prove Lemma \ref{creins}, which we restate here for convenience.
\creins*
\begin{proof}
We first use Lemma \ref{initialize} to create a collection of states $(\mathcal{P}^{i})_{i \in[k_0]}$ for some $k_0=f(\pi,t)$ such that one of them is maximum valid and $\tau(\mathcal{P}^{i})\leq 2 \pi+t$ holds for all $i \in [k_0]$. Now, throughout the algorithm, we maintain a multicut $S$ for $(G,H)$ and a collection $\mathcal{I}$ of extended subinstances $(S',(G',H'))$ of $(G,H)$ satisfying $\pi(G')+|V(H')|<\pi(G)+|V(H)|$. We initialize $S$ by $E(G)$ and $\mathcal{I}$ by $\emptyset$. It follows from Lemma \ref{crzwei} that there exists an algorithm running in $f(\pi,t)n^{O(\sqrt{\mu})}$ that we can apply to $(G,H)$ and a state $\mathcal{P}$ and that, in case that $\mathcal{P}$ is maximum valid, returns a subset of $E(G)$ and a collection of extended subinstances of $(G,H)$ such that either the edge set is a minimum multicut for $(G,H)$ or one of the extended subinstances is optimumbound. We now apply this algorithm to $(G,H)$ and $\mathcal{P}^{i}$ for all $i \in [k_0]$. Whenever the output contains an edge set $S'$ with $|S'|<|S|$, we replace $S$ by $S'$ and whenever the output contains a collection of extended subinstances of $(G,H)$, we add all these extended subinstances to $\mathcal{I}$. We make the algorithm return the assignments of $S$ and $\mathcal{I}$ after the last one of these iterations.

For the correctness of the algorithm, consider some $i \in [k_0]$ such that $\mathcal{P}^{i}$ is maximum valid. We obtain that when we apply the algorithm of Lemma \ref{crzwei} to $(G,H)$ and $\mathcal{P}^{i}$, we obtain an output that contains either a minimum multicut $S'$ for $(G,H)$ or an optimumbound extended subinstance of $(G,H)$. In the first case, we obtain that $S$ is assigned a minimum multicut for $(G,H)$ after this iteration and in the second case, we obtain that an obtimumbound extended subinstance of $(G,H)$ is added to $\mathcal{I}$ in this iteration. Finally observe, as $k_0=f(\pi,t)$ and for every $i \in [k_0]$, we add $f(\pi,t)$ extended subinstances of $(G,H)$ to $\mathcal{I}$, we have $|\mathcal{I}|=f(\pi,t)$ after the last one of these iterations.

For the running time of the algorithm, first observe that the algorithm of Lemma \ref{initialize} runs in $f(\pi,t)n^{O(1)}$ by assumption. Next observe that for every $i \in [k_0]$, the algorithm of Lemma \ref{crzwei} runs in $f(\pi,t)n^{O(\sqrt{\mu})}$. Further, we can clearly check in polynomial time if the output is of the desired form. As $k_0=f(\pi,t)$, we obtain that the entire algorithm runs in $f(\pi,t)n^{O(\sqrt{\mu})}$.
\end{proof}

\subsection{The final algorithm}\label{finalg}
We are now ready to prove Theorem \ref{extraedges}, which we restate here for convenience. While this involves the careful creation of a search tree, again, the proof conecept is rather straight forward relying on \cref{creins}.

\extraedges*

\begin{proof}
In a first phase of the algorithm, we create a rooted search tree such that each node $q$ of the search tree is associated to an extended subinstance $(S_q^1,(G_q,H_q))$ of $(G,H)$ and a multicut $S_q^2$ of $(G_q,H_q)$. We start by creating the root $r$ of this search tree and defining $S_r^1=\emptyset$ and $(G_r,H_r)=(G,H)$. In the following, while there exists a node $q$ that has not been considered after its creation and is associated to an instance $(S_q^1,(G_q,H_q))$ with $H_q$ not being edgeless, we compute $S_q^2$ and create its children together with the extended subinstance of $(G,H)$ associated to them. 

Now consider a node $q$ that has not been considered after its creation. First suppose that $G_q$ is disconnected and let $G_q^1,\ldots,G^k_q$ be the connected components of $G_q$ that satisfy $V(G^{i}_q)\cap (W \cup T)\neq \emptyset$. Observe that $k \leq |W \cup T|=2\pi+t$. We now add a collection of $k$ children $q_1,\ldots,q_k$ of $q$ to the search tree. For $i \in [k]$, we set $S_{q_i}^1=\emptyset$ and $(G_{q_i},H_{q_i})=(G_q^{i},H[V(G_q^{i})])$. We further set $S_q^2=E(G_q)$.

Now suppose that $G_q$ is connected. We first use the algorithm of Lemma \ref{findmod} to compute a set $E'_q\subseteq E(G_q)$ with $|E'_q|=\pi(G_q)$ such that $G_q \setminus E'_q$ is planar. From now on, we may assume that $(G_q,H_q)$ is a $\kplanar$ instance.  By Lemma \ref{creins} and as $(G_q,H_q)$ is a subinstance of $(G,H)$, in  $f(\pi,t)n^{O(\sqrt{\mu})}$, we can run an algorithm that returns a set $S'$ and a collection of extended subinstances $(S_i,(G_i,H_i))_{i \in [k]}$ of $(G,H)$ for some $k=f(\pi,t)$ such that for all $i \in [k]$, we have that $\pi(G_i)+|V(H_i)|<\pi(G_q)+|V(H_q)|$ holds and either $S'$ is a minimum multicut for $(G_q,H_q)$ or there exists some $i \in [k]$ such that $(S_i,(G_i,H_i))$ is an optimumbound extended subinstance of $(G_q,H_q)$. We first set $S_q^2=S'$. Then, for $i \in [k]$, we create a child $q_i$ of $q$ and set $(S_{q_i},(G_{q_i},H_{q_i}))=(S_i,(G_i,H_i))$. 

We do this as long as there exists a node $q$ that has not been considered after its creation and is associated to an extended subinstance $(S_q^1,(G_q,H_q))$ of $(G,H)$ with $H_q$ not being edgeless. Finally, for all nodes $q$ which are associated to an extended subinstance $(S_q^1,(G_q,H_q))$ of $(G,H)$ with $H_q$ being edgeless, we set $S_q^2=\emptyset$. This finishes the description of the search tree, whose assignment in the end of the procedure we denote by $U$. This finishes the description of the first phase of the algorithm.

Observe that for every node $q$ of $U$ and every child $q'$ of $q$ in $U$, we have $\pi(G_{q'})+|V(H_{q'})|\leq \pi(G_{q})+|V(H_{q})|$. Moreover, for every grandchild $q'$ of $q$, we have that $\pi(G_{q'})+|V(H_{q'})|< \pi(G_{q})+|V(H_{q})|$ is satisfied. It follows that the depth of $U$ is bounded by $2(\pi+t)$. Further, by the assumption on the algorithm of Lemma \ref{creins} and construction, we have that every node of $U$ has $f(\pi,t)$ children. It follows that $U$ has $f(\pi,t)$ nodes in total. Now, for every node of $U$, we have that the running time associated to this node is $f(\pi,t)n^{O(\sqrt{\mu})}$ by the assumption on the algorithms of Lemmas \ref{creins} and \ref{findmod}, and as the connected components of a graph can clearly be computed in polynomial time. It follows that the first phase of the algorithm runs in $f(\pi,t)n^{O(\sqrt{\mu})}$.

In the second phase of the algorithm, we compute a minimum multicut $S^*_q$ for $(G_q,H_q)$ for every node $q$ of $U$ by a bottom-up approach. First, for every leaf $q$ of $U$ which is associated to an extended subinstance $(S_q^1,(G_q,H_q))$ of $(G,H)$ with $H_q$ being edgeless, we set $S^*_q=\emptyset$. Clearly, we have that $\emptyset$ is a minimum multicut for $(G_q,H_q)$, so this assignment is correct. Now consider a node $q$ which is associated to an extended subinstance $(S_q^1,(G_q,H_q))$ of $(G,H)$ with $H_q$ not being edgeless and let $q_1,\ldots,q_k$ be the children of $q$.  Recursively, we may suppose that $S_{q_i}^*$ has already been computed for $i \in[k]$. 

First suppose that $G_q$ is disconnected. In that case, let $S_q^*=\bigcup_{i \in [k]}S_{q_i}^*$, where $q_1,\ldots,q_k$ are the children of $q$. It follows directly from the construction of $U$ that $S_q^*$ is indeed a minimum multicut for $(G_q,H_q)$.

Now suppose that $G_q$ is connected. We now let $\mathcal{S}_q$ contain $S_q^2$ and $S_{q_i}^*\cup S_{q_i}^1$ for $i \in [k]$. We now let $S_q^*$ be the smallest multicut for $(G_q,H_q)$ that is contained in $\mathcal{S}_q$. It follows from the assumption on the algorithm of Lemma \ref{creins} that $S_q^*$ is indeed a minimum multicut for $(G_q,H_q)$. We now recursively compute $(G_q,H_q)$ for all nodes $q$ of $U$. In the end, we have in particular computed $S_r^*$, which is a minimum multicut for $(G,H)$. We output $S_r^*$.

It is not difficult to see that for every node $q$ of $U$, we can compute $S_q^*$ in $f(\pi,t)n^{O(1)}$ once $S_q^2$ and $S_{q'}^*$ for all children $q'$ of $q$ are computed. As the number of nodes of $U$ is $f(\pi,t)$, we obtain that the second phase of the algorithm runs in $f(\pi,t)n^{O(1)}$. Hence the overall running time of the algorithm is $f(\pi,t)n^{O(\sqrt{\mu})}$.
\end{proof}

\section{Crossing number}

In this section, a {\it drawing} of a graph $G$ consists of a mapping of the vertices of $G$ to points in the plane and a mapping of the edges of $G$ to curves in the plane connecting their endpoints such that all vertices and edges of $G$ are in general position except the required adjacencies and such that in any point not corresponding to a vertex, at most two edges intersect. An {\it edge crossing} refers to a point that does not correspond to a vertex and in which two edges intersect. Given an instance $(G,H)$ of \mc\, we use $cr$ for the {\it crossing number} of a graph $G$, which is defined to be the smallest integer $k$ such that there exists a drawing of $G$ with exactly $k$ edge crossings. The following is the main result of this section.
\maincr*
 In order to make our algorithm work, rather than just the crossing number of a graph, we also need to have an optimal drawing available. This can be achieved through the following result due to Grohe \cite{GROHE2004285}.

\begin{prop}\label{grohecr}
Given a graph $G$, there exists an algorithm that computes a drawing of $G$ in the plane minimizing the number of crossings and that runs in time $f(cr)n^{O(1)}$.
\end{prop}

Formally, an instance $(G,H)$ of \mc\ is called \emph{crossing-planar} if $G$ is given as a drawing of a graph. Given a crossing-planar instance, we denote by $E_{cr}$ the set of edges in $E(G)$ crossing at least one other edge of $E(G)$ and  we set $\overline{cr}=|E_{cr}|$. Observe that $\overline{cr}$ can be arbitrarily large in comparison to $cr$ as $\overline{cr}$ refers to the given drawing and $cr$ to an optimal drawing of the abstract graph $G$. However, when considering an optimal drawing of $G$, we have $\overline{cr}\leq2cr$ as every crossing involves 2 edges.

For our algorithm, it will be convenient to assume that the instances in consideration satisfy a technical extra condition.  Namely, a crossing-planar instance $(G,H)$ of \mc\ is called \emph{normalized} if every $e \in E_{cr}$ is of infinite weight.
Observe that for any normalized crossing-planar instance $(G,H)$ of \mc, any finite multicut $S$ of $(G,H)$ and any $e \in E_{cr}$, we have that $V(e)$ is contained in one component of $G\setminus S$.
Our main technical contribution is proving the restriction of Theorem \ref{maincr} to normalized instances.

\begin{thm}\label{maincrnorm}
There exists an algorithm that solves every normalized crossing-planar instance of \mc\ and runs in $f(\overline{cr},t)n^{O(\sqrt{\mu})}$. 
\end{thm}

Before the more involved proof of \cref{maincrnorm}, we first show that \cref{maincrnorm}. \begin{proof}[Proof (of \cref{maincr})]
By Proposition \ref{grohecr} and as the running time in Proposition \ref{grohecr} is dominated by the running time in Theorem \ref{maincr}, we may suppose that $(G,H)$ is given in the form of a crossing-planar instance with $\overline{cr}\leq 2 cr$.
Consider some $Z \subseteq E_{cr}$. Then we let $G_Z$ be obtained from $G$ by deleting $Z$ and assigning all edges in $E_{cr}\setminus Z$ infinite weight. Observe that $(G_Z,H)$ is a normalized crossing-planar instance for every $Z \subseteq E_{cr}$. Using Theorem \ref{maincrnorm}, we can hence compute a minimum weight multicut $S_Z$ for $(G_Z,H)$ and check whether $S_Z \cup Z$ is a multicut for $(G,H)$. We let $S$ be the minimum weight multicut for $(G,H)$ found during this procedure. In order to see that $S$ is a minimum weight multicut for $(G,H)$, let $S^*$ be a minimum weight multicut for $(G,H)$ and let $Z^*=S^*\cap E_{cr}$. Then $S^*\setminus Z^*$ is a multicut for $(G_{Z^*},H)$ and hence a minimum weight multicut for $(G,H)$ is found when considering $Z=Z^*$. Finally, we execute the algorithm of \cref{maincrnorm} only $2^{\overline{cr}}$ times and hence the desired running time follows.
\end{proof}
It remains to prove \cref{maincrnorm}.

Given a normalized crossing-planar instance $(G,H)$ of \mc , we denote by $\mathcal{H}$ the set of graphs $H'$ on $V(H)\cup V(E_{cr})$ and we use $G'$ for $G\setminus E_{cr}$. Further, we let $F^*$ denote the faces of $G'$ which contain a pair of crossing edges in $G$. Recall that $C-F^*$ is the graph obtained from $C$ by removing every vertex that is in a face of $F^*$ and removing every edge that intersects a face of $F^*$. 
The main technical difficulty for the proof of \cref{maincrnorm} is contained in the following result.
\begin{restatable}{lem}{gvugzgzu}\label{gvugzgzu}
Let $(G,H)$ be a normalized crossing-planar instance of \mc{}   admitting a multicut of finite weight. Then there exists \omitnow{some }$H' \in \mathcal{H}$ such that every minimum weight multicut for $(G',H')$ is a minimum weight multicut for $(G,H)$ and for every subcubic, inclusion-wise minimal minimum weight multicut dual $C$ for $(G',H')$, we have $\tw(C-F^*)=O(\sqrt{\mu})$. 
\end{restatable}
\begin{proof}
Let $(B_1,B_2,I,X)$ be an extended biclique decomposition of $H$ with $|X|=\mu$. Further, let $S^*$ be a minimum weight multicut for $(G,H)$ such that $S^*$ minimizes the number of elements of $V(H)\cup V(E_{cr})$ that are in a component of $G\setminus S^*$ that also contains a terminal of $X$. In the following, we use $\overline{X}$ for the set of elements of $V(H)\cup V(E_{cr})$  that are in a component of $G\setminus S^*$ that also contains a terminal of $X$,  we use ${\overline{B}_1}$ for the set of elements of $(V(H)\cup V(E_{cr}))\setminus \overline{X}$  that are in a component of $G\setminus S^*$ that also contains a terminal of $B_1$, and we use ${\overline{B}_2}$ for $(V(H)\cup V(E_{cr}))\setminus (\overline{X}\cup {\overline{B}_1})$. (Here we have a slight asymmetry between the definitions of ${\overline{B}_1}$ and ${\overline{B}_2}$: if a component of $G\setminus S^*$ contains only terminals from $I$, then they are classified to be in ${\overline{B}_1}$. However, this asymmetry will be inconsequential in the proof.) Observe that $({\overline{B}_1},{\overline{B}_2},\overline{X})$ is a partition of $V(H)\cup V(E_{cr})$. We now define $H'$ to be the complete multipartite graph such that ${\overline{B}_1}$ and ${\overline{B}_2}$ are partition classes and such that for every component of $G\setminus S^*$ that contains at least one terminal of $X$, all the elements of $\overline{X}$ contained in this component form a partition class.


 This finishes the description of $H'$. Observe that $H' \in \mathcal{H}$. Further observe that $H$ is a subgraph of $H'$. Indeed, consider some $t_1,t_2 \in V(H)$ with $t_1t_2 \in E(H)$. If $\{t_1,t_2\}\cap \overline{X}=\emptyset$, we have $\{t_1,t_2\}\cap X=\emptyset$ and may hence suppose by symmetry that $t_i \in B_i$ for $i \in [2]$. It follows by construction and $S^*$ being a multicut for $(G,H)$ that $t_i \in {\overline{B}_i}$ for $i \in [2]$ yielding $t_1t_2 \in E(H')$ by construction. On the other hand, if one of  $t_1$ and $t_2$ is contained in $\overline{X}$, then $t_1$ and $t_2$ are in distinct components of $G\setminus S^*$ as $S^*$ is a multicut for $(G,H)$ and so $t_1t_2 \in E(H')$ follows by construction.

In the following, we prove through two claims that every minimum weight multicut for $(G',H')$ is also a minimum weight multicut for $(G,H')$. We first give the following result which will later be reused in the proof of the treewidth bound.
\begin{restatable}{clm}{crossa}\label{fghuj}
Let $S'$ be a multicut for $(G',H')$. Then $S'$ is also a multicut for $(G,H')$. 
\end{restatable}
\begin{proof}
Let $t_1,t_2 \in V(H')$ such that there exists a path $P$ from $t_1$ to $t_2$ in $G\setminus S'$. Let $v_0,\ldots,v_k$ be the vertices of $V(P)\cap V(H')$ in the order they appear on $P$ when going from $t_1$ to $t_2$ with $v_0=t_1$ and $v_k=t_2$. As $H'$ is a complete multipartite graph, it suffices to prove  that $H'$ does not contain an edge linking $v_{i-1}$ and $v_i$ for any $i \in [k]$: this shows that there cannot be an edge between $v_0$ and $v_k$. 

Let $i \in [k]$.  If the subpath of $P$ from $v_{i-1}$ to $v_i$ contains an edge of $E_{cr}$, then (as $V(E_{cr})\subseteq V(H')$, we actually have $v_{i-1}v_i \in E_{cr}$. In this case, we obtain that $v_{i-1}v_i  \notin S^*$ as the weight of $v_{i-1}v_i$ is infinite and the weight of $S^*$ is finite. It follows that $v_{i-1}$ and $v_i$ are in the same connected component of $G\setminus S^*$. By the choice of $\overline{X}, {\overline{B}_1}$, and ${\overline{B}_2}$, we obtain that $v_{i-1}$ and $v_i$ are not linked by an edge of $E(H')$. Consider now the case when the subpath of $P$ from $v_{i-1}$ and $v_i$ does not use an edge of $E_{cr}$. Then $v_{i-1}$ and $v_i$ are in the same connected component of $G'\setminus S'$. As $S'$ is a multicut for $(G',H')$, we obtain that $t_1$ and $t_2$ are not linked by an edge of $E(H')$.

Hence $S'$ is indeed a multicut for $(G,H')$. 
\claimqedhere\end{proof}

In order to prove the optimality, we further need the following result.
\begin{restatable}{clm}{crossb}\label{drzftgzuhui}
$S^*$ is a multicut for $(G',H')$.
\end{restatable}
\begin{proof}
Let $t_1,t_2 \in V(H')$ with $t_1t_2 \in E(H')$. If $\{t_1,t_2\}\cap \overline{X} \neq \emptyset$, we obtain that $t_1$ and $t_2$ are in distinct components of $G\setminus S^*$ by the definition of $H'$. As $G'\setminus S^*$ is a subgraph of $G\setminus S^*$, we obtain that $t_1$ and $t_2$ are in distinct components of $G'\setminus S^*$.

Now suppose that $\{t_1,t_2\}\subseteq {\overline{B}_1} \cup {\overline{B}_2}$. By the definition of $H'$ and by symmetry, we may suppose that $t_i \in {\overline{B}_i}$ for $i \in [2]$. 
It follows by definition of ${\overline{B}_1}$ and ${\overline{B}_2}$ that $t_1$ is in a component of $G\setminus S^*$ containing a terminal of $B_1$ and $t_2$ is in a component of $G\setminus S^*$ not containing a terminal of $B_1$. In particular, $t_1$ and $t_2$ are in distinct components of $G\setminus S^*$.    As $G'\setminus S^*$ is a subgraph of $G\setminus S^*$, we obtain that $t_1$ and $t_2$ are in distinct components of $G'\setminus S^*$.
\claimqedhere\end{proof}

We are now ready to show that any minimum weight multicut $S$ for $(G',H')$ is a minimum weight multicut for $(G,H)$. First  observe that $S$ is a multicut for $(G,H)$ by Claim \ref{fghuj} and as $H$ is a subgraph of $H'$. Further, by Claim \ref{drzftgzuhui} and the minimality of $S$, we obtain $w(S)\leq w(S^*)$ where $w$ is the weight function associated to $G$. By the minimality of $S^*$, the statement follows.

\medskip

In the following, let $C$ be a subcubic, inclusion-wise minimal, minimum weight multicut dual for $(G',H')$.  We will give the desired treewidth bound on $C$. Let $S'=E_{G'}(C)$ and observe that $S'$ is a minimum weight multicut for $(G',H')$ by \cref{dualcut}. It hence follows from the first part of the lemma that $S'$ is also a minimum weight multicut for $(G,H)$. In the following, we use $C'$ for $C-F^*$.

We first need the following result.
\begin{restatable}{clm}{crossc}\label{rdzthzujiop}
Let $f$ be a face of $C'$ that contains a terminal of $\overline{X}$. Then $f$ also contains a terminal of $X$.
\end{restatable}

\begin{proof}
Suppose otherwise. We will obtain a contradiction to the fact that $S^*$ is chosen so that the number of elements of $\overline{X}$ is minimized.  Let $\overline{X}'$ denote the set of terminals in $V(H')$ that are contained in a component of $G\setminus S'$ that also contains a terminal of $X$. As $f$ is a face of $C'$, for every every $e \in E_{cr}$, either both or none of its endvertices are contained in $f$ and hence there is a collection of components of $G\setminus S'$ whose union is the set of vertices of $V(G)$ contained in $f$.  By assumption, we obtain that $\overline{X}\setminus \overline{X}'\neq \emptyset$.

Next consider some $x_0 \in \overline{X}'$. Let $K$ be the component of $G\setminus S'$ containing $x_0$. By the definition of $\overline{X}'$, we obtain that there exists some $x_1 \in V(K)\cap X$. By Claim \ref{fghuj}, we obtain that $S'$ is a multicut for $(G,H')$ and hence $x_0$ and $x_1$ are not linked by an edge of $E(H')$. By the definition of $H'$, we obtain that $x_0$ and $x_1$ are in the same component of $G\setminus S^*$. In particular, we have $x_0 \in \overline{X}$. This yields $\overline{X}'\subseteq \overline{X}$.

We obtain $|\overline{X}'|<|\overline{X}|$. As $S'$ is a minimum weight multicut for $(G,H)$, we obtain a contradiction to the choice of $S^*$.
\claimqedhere\end{proof}

The next result is a crucial observation for the treewidth bound for $C$.
\begin{restatable}{clm}{crossd}\label{zftzuguhi}
 Let $v \in V(C')$ with $d_{C'}(v)=3$. Then $v$ is incident to a face of $C'$ that contains a terminal of $\overline{X}$.
 \end{restatable}
 \begin{proof}
 As $C$ is subcubic and $C'$ is a subgraph of $C$, we have $d_C(v)=3$.
 Let $f_1,f_2$, and $f_3$ be the faces of $C$ incident to $v$. Observe that $f_1,f_2$, and $f_3$ are well-defined and pairwise distinct by the minimality of $C$. Now suppose for the sake of a contradiction that none of $f_1,f_2$, and $f_3$ contain a terminal of $\overline{X}$. Observe that, as $C$ is a multicut dual for $(G',H')$ for $i \in [3]$, there exists some $j \in [2]$ such that all the terminals of $V(H')$ contained in $f_i$ are contained in $\overline{B}_j$. By symmetry, we may suppose that there exists some $j \in [2]$ such that all the terminals contained in $f_1$ or $f_2$ are contained in $\overline{B}_j$. Let $e \in E(C)$ be the edge separating $f_1$ and $f_2$. We obtain that $C\setminus \{e\}$ is a multicut dual for $(G',H')$. This contradicts the minimality of $C$. 
 \claimqedhere\end{proof}
 We are now ready to conclude the desired treewidth bound on $C'$.
 Let $C''$ be a component of $C'$. 
 For every face $f$ of $C''$, we now introduce a vertex $z_f$ that shares an edge with every vertex that is incident to $f$. Let $C''_0$ be the resulting graph.
Observe that $C''_0$ is planar. Now let $A\subseteq V(C''_0)$ contain the vertex $z_f$ for all faces $f$ of $C''$ which contain a terminal of $X$. If $d_{C''}(v)\leq 2$ for all $v \in V(C'')$, we clearly have $\tw(C'')\leq 2$. Otherwise, it follows directly from Claims \ref{rdzthzujiop} and \ref{zftzuguhi} and the fact that the faces of $C''$ correspond to the faces of $C'$ that $A$ is a 3-dominating set of $C''_0$: if a face $f'$ of $C''$ is adjacent to a face $f$ of $C''$ containing a vertex from $X$, then $z_{f'}$ is at distance $2$ from $z_f$, and consequently every vertex incident to $f'$ is at distance at most $3$ from $z_f$. 
By Proposition \ref{domplanar} and $|A|\leq |X|\leq \mu$, we obtain that $\tw(C''_0) =O(\sqrt{\mu})$.  As $C''$ is a subgraph of $C''_0$, we obtain that $\tw(C'') =O(\sqrt{\mu})$. As $C''$ was chosen to be an arbitrary component of $C'$, this yields $\tw(C') =O(\sqrt{\mu})$. 
\end{proof}

With \cref{gvugzgzu} at hand, we are now ready to conclude \cref{maincrnorm}.

\begin{proof}[Proof (of \cref{maincrnorm})]
By Proposition \ref{grohecr} and as the running time in Proposition \ref{grohecr} is dominated by the running time in Theorem \ref{maincr}, we may suppose that $(G,H)$ is given in the form of a normalized crossing-planar instance with $\overline{cr}\leq 2 cr$.

By Lemma \ref{gvugzgzu}, there exists some $H'_0 \in \mathcal{H}$ such that every minimum weight multicut for $(G',H'_0)$ is a minimum weight multicut for $(G,H)$ and for every subcubic, inclusion-wise minimal minimum weight multicut dual $C$ for $(G',H'_0)$, we have $\tw(C-F^*)=O(\sqrt{\mu})$. 
Now consider some $H' \in \mathcal{H}$. By Lemma \ref{algogenext}, there exists an algorithm that runs in $f(|V(H')|)n^{O(\sqrt{\mu})}$ and returns a minimum multicut for $(G,H)$ if $H'=H'_0$. We apply this algorithm to $(G',H')$ for every $H' \in \mathcal{H}$ and output the minimum multicut for $(G,H)$ we find during this procedure.

The correctness of the algorithm follows directly from the fact that $H'_0\in \mathcal{H}$. For the running time of the algorithm, observe that $|\mathcal{H}|=f(\pi,t)$. As $|V(H')|=f(\pi,t)$ for all $H' \in \mathcal{H}$, it follows that the entire algorithm runs in $f(\pi,t)n^{O(\sqrt{\mu})}$.
\end{proof}

\label{crosssec}

\bibliographystyle{plainurl}
\bibliography{refs}




\end{document}